\documentclass[preprint]{article}
\usepackage{graphicx}
\usepackage{amsmath}
\usepackage{arydshln}
\usepackage{amssymb}
\usepackage{rotating}
\usepackage[hidelinks]{hyperref}
 \usepackage{amsthm}
 \usepackage{subcaption}

\usepackage{libertine}

\newcommand{\toprowspace}{\rule{0pt}{4ex}}
\newcommand{\midrowspace}{\rule{0pt}{2ex}}

\usepackage{longtable}
 
\usepackage{stmaryrd}

\usepackage{booktabs}
\usepackage{collcell}
\newcommand*{\movedown}[1]{%
  \smash{\raisebox{-1ex}{#1}}}
\newcolumntype{q}{>{\collectcell\movedown}r<{\endcollectcell}}

\usepackage{array}

 \usepackage{comment}

 \usepackage{xcolor}
\newtheorem{cor}{Corollary}[section]
\newtheorem{prop}{Proposition}[section]
\newtheorem{conj}{Conjecture}[section]
\newcommand*{\img}[1]{%
    \raisebox{-.3\baselineskip}{%
        \includegraphics[
        height=\baselineskip,
        width=\baselineskip,
        keepaspectratio,
        ]{#1}%
    }%
}
\usepackage{wrapfig}
\graphicspath{ {./images/} }

\newcommand{\be}{\begin{equation}}
\newcommand{\ee}{\end{equation}}
\newcommand{\ba}{\begin{eqnarray}}
\newcommand{\ea}{\end{eqnarray}}
\newcommand{\la}{\lambda}
\newcommand{\al}{\alpha}

\newcommand{\tr}{\rm tr}
\newcommand{\g}{\gamma}

\begin{document}
\setlength\dashlinedash{.6pt}
\setlength\dashlinedash{2pt}
\hoffset=-.4truein\voffset=-0.5truein
\setlength{\textheight}{8.5 in}
\begin{titlepage}
\begin{center}
\vskip 40 mm
{\large \bf Braid  twists and HZ factorisation  via   character expansion 
}

\vskip 10mm
{\bf {Andreani Petrou and Shinobu Hikami }}
\vskip 5mm

Okinawa Institute of Science and Technology Graduate University,\\
 1919-1 Tancha, Okinawa 904-0495, Japan.

\vskip 5mm
{\bf Abstract}
\vskip 3mm
\end{center}
The HOMFLY--PT polynomial of a closed braid 
admits a character expansion, expressed  as a sum  of $SU(N)$ characters 
over Young diagrams. The Harer--Zagier (HZ) transform, which 
converts the HOMFLY--PT polynomial into a rational function, is applied directly to the characters, yielding  the   \emph{HZ character expansion}. 
We use this expansion 
to illumine the hidden 
structure of the HZ function that enables its decomposition into a sum of factorised terms.   
We further articulate explicit conditions for  full HZ factorisation, which include that non-vanishing contributions come solely from hook-shaped Young diagrams. 
We show that these conditions remain invariant under three mutually commuting braid twists: full twists, partial full twists and Jucys-Murphy twists. 
We, hence, employ such twists 
to  construct   infinite,  HZ-factorisable families  of 
knots and links, which  can often be thought of as a hyperbolic extension of
 torus knots. Remarkably, these families encompass the Coxeter links corresponding to $E$-type Dynkin diagrams.



 \end{titlepage}
 \tableofcontents

\section{Introduction}
Knots and their invariants  are 
important  topological objects in three dimensions that
  have attracted increasing attention by both mathematicians and physicists for several decades. 
 The seminal paper by Witten \cite{Witten0} 
 has revolutionised our understanding of knot invariants, not only by providing them with a physics interpretation in terms of Chern--Simons gauge theory, but also by making  representation theoretic tools accessible to their study. Namely,
the  HOMFLY--PT polynomial 
of a knot or link, which is a two-variable generalisation of both the Jones  and Alexander polynomials,   can be described by the  gauge group $SU(N)$; and its variables can be  expressed in terms of the rank $N$   and the quantum group parameter  $q$ \cite{Reshetikhin}.  

 In the series of works \cite{MorozovCH,MorozovCHII,Anokhina},
 the authors show  that the  HOMFLY--PT polynomial $\bar H(\mathcal{K})$ of a knot $\mathcal{K}$, obtained as the closure of an  $m$-strand braid, 
 can be expanded in terms of   characters of $SU(N)$, i.e.  Schur functions $\hat{S}_Q$, as
 \be\label{Racah1}
 \bar{H}(\mathcal{K};q^N,q) = \sum_Q h^Q(q) \hat{
 S}_Q(q^N,q).
\ee
 Here the sum is over all possible Young diagrams $Q$ with $m$ boxes and   $h^Q$ are called  the Racah coefficients. The latter  are evaluated by the trace of a product of $R-$matrices and 
 have an interesting interpretation as 6-j symbols \cite{KauffmanLins,Kirillov}. 

The Harer--Zagier (HZ) transform is a discrete  Laplace transform that can be applied to the (unnormalised\footnote{The standard  HOMFLY-PT polynomial $H(\mathcal{K})$, normalised such that $H(\bigcirc)=1$, is related to (\ref{Racah1}) by 
  $  H(\mathcal{K})=\frac{z}{A-A^{-1}}\bar{H}(\mathcal{K})$.})  HOMFLY--PT polynomial and illuminates some of its topological content, such as the lower bound for the braid index \cite{Petrou1}. 
It can be  thought of as a generating function, denoted by $Z(\mathcal{K})$,  defined as \cite{Morozov}
\be\label{HZdef}
Z(\mathcal{K};\lambda,q) = \sum_{N=0}^\infty \bar{H}(\mathcal{K}; q^N,q) \lambda^N .
\ee
 It is evaluated,  via the geometric series, by  the substitution
  \be\label{HZ2}
 q^{Nk}\to(1- \lambda q^k)^{-1}.
 \ee 
 Hence, its effect is to transform a polynomial in $q^N$ into a rational function that involves the ratio of polynomials in the parameters $q$ and $\lambda$. 
If both the numerator and denominator of this ratio can be expressed as the product of monomials in $\lambda$ of the form $(1\pm \lambda q^k)$,  it is said to be \emph{(fully) factorisable}. HZ factorisability is an important property, which was known to hold for torus knots \cite{Morozov}, and has been conjectured  in \cite{Petrou1} to be equivalent with a relation between the HOMFLY--PT and Kauffman polynomials. 
In order to better understand this property, it is instructive to apply the HZ transform  to 
the character expansion (\ref{Racah1}) of the HOMFLY--PT polynomial. Since the dependence on $q^N$  appears only through the Schur functions, the HZ character expansion can be expressed as
 \be\label{CHE}
Z(\mathcal{K};\lambda,q)= \sum_Q h^{Q}(q) Z(\hat S_Q;\lambda,q).
 \ee

In our previous articles \cite{Petrou} and \cite{Petrou1}, often referred to hereafter as I  and II, respectively, the  HZ transform and its  factorisability properties were 
investigated via experimental mathematics and several hyperbolic families of knots and links admitting HZ factorisability were discovered. Such families were generated by concatenating a special `base braid` with full twists and Jucys--Murphy twists. In the present work, the HZ character expansion  in (\ref{CHE}) is used to  rigorously prove the HZ factorisability of such families and their generalisations, by showing that the twisting operations used to generate them correspond to diagonal elements in the $R-$matrix algebra. These preserve the (sufficient) conditions  for HZ factorisability, which are expressed in terms of the Racah coefficients $h^Q$ and  demand vanishing contributions from non-hook shaped Young diagrams. 

In the non-factorisable cases, which involve the vast majority of knots,  the character expansion  and, in particular, the structure of the Racah coefficients $h^Q$ can be used to decompose the HZ function  of an arbitrary closed braid into a sum of factorised terms. 
 As a simple example, we consider  the HZ expansion of the figure-8 knot, obtained as the closed $3$-strand braid $(\sigma_1\sigma_2^{-1})^2$,  which reads
 \ba\label{Z4_1D}
 Z(4_1;\lambda,q)&=& \frac{1}{\mathcal{D}_3}\biggl(
 \lambda h^{[3]}+ (q+q^{-1})h^{[21]} \lambda^2 + \lambda^3 h^{[111]}\biggr)\nonumber\\
 &=&\frac{1}{\mathcal{D}_3}\biggl(
 \lambda + (q^5+q^{-5}-q^3-q^{-3}-q-q^{-1})\lambda^2 + \lambda^3 \biggr)\nonumber\\
 &=& -[-5,5]+[-3,3]+[-1,1].
 \ea
 Here $\mathcal{D}_3=(1-q^{-3}\lambda)(1-q^{-1}\lambda)(1-q \lambda)(1-q^3\lambda)$, 
while the brackets in the last line are defined by  $[-n,n]:= \frac{\lambda}{\mathcal{D}_3}(1-q^{-n}\lambda)(1-q^n\lambda)$. Since these brackets have a factorised form, the final expression in (\ref{Z4_1D}) will be called a \textit{factorised form decomposition}. 
We will give examples of such decompositions for several closed braids with up to $ 5$ strands in
 Sec.~\ref{sec:factform}. While we provide an algorithmic approach to obtain the factorised form decomposition that is valid for closed braids with up to $8$ strands,  its existence can only be  conjectured more generally. Its applicability to links with two components is shown  in Sec.~\ref{sec:forestquiv}, through examples involving D-type Coxeter links.

\paragraph{Organisation of the paper.}
In Sec.~\ref{sec:CharExp}, the   character expansion for the HOMFLY-PY polynomial and its HZ transform for closed braids with up to $6$ strands  is explained through explicit computations. The HZ factorisability conditions  are provided and their persistency under the braid operations of full twist, partial full twist and Jucys-Murphy twist is proven in each case.
These twists, applied $j,k$ and $l$ times, respectively, are used to construct the hyperbolic family of knots $\mathcal{K}^{(m)}_{j,k,l}$, which is  extrapolated to  arbitrary number of strands $m$ and can be thought of as a hyperbolic extension of torus knots $T(m,n)$. The character expansion of the Jones and Alexander polynomials is also considered at the end of Sec.~\ref{sec:CharExp}. 
 The factorised form decomposition of the HZ function for general closed braids, as in (\ref{Z4_1D}), is outlined in Sec.~\ref{sec:factform}.  In Sec.~\ref{sec:forestquiv}, Coxeter links  corresponding to Dynkin diagrams of ADE type are considered via the HZ transform, while 
  Sec.~\ref{sec:summary} is devoted to a summary and discussions. The Appendix lists the HZ function, in the form of a factorised form decomposition, 
  for all knots in the Rolfsen table with up to $8$ crossings and for selected $9$-crossing knots. 

\section{The HOMFLY--PT character expansion and its HZ transform}\label{sec:CharExp}

The  HOMFLY--PT polynomial of a closed braid with $m$-strands and writhe $w$ admits the character expansion   \cite{MorozovCH,MorozovCHII}
\be\label{Schur}
 \bar H(\mathcal{K};A,q) = A^{-w}\sum_{Q} h^Q(q) S_Q(A,q),
\ee
where $S_Q$ is a Schur function and $Q$ is a partition of $m$, which is represented by a Young diagram with $m$ boxes.  This gives an alternative  to the  skein relation  method (used in\footnote{The variables $(a,z)$ appearing in I and II are related to $(A,q)$ by $A=a^{-1}$ and $z=q-q^{-1}$.\label{foot:A=a^-1}
} I and II) to derive the HOMFLY--PT polynomial of a knot.  
The coefficients $h^Q$, which  
 are called the
Racah coefficients, are functions that only depend on $q$ and are equivalent to the Wigner 6-j symbols  
\cite{Reshetikhin,And,KauffmanLins}. 
For an $m$-strand braid  $ \prod_{i=1}^\infty \prod_{j=1}^{m-1}  \sigma_j^{a_{ij}}\allowbreak \in B_m$, where $a_{ij}\in\mathbb{Z}$ and   $\sigma_i$ are the generators of the braid group $B_m$, 
$h^Q$ can be explicitly computed as the trace of the product of quantum $R-$matrices\footnote{These are a quantum version of the classical $R$-matrices, which solve the Yang-Baxter equation \cite{Kauffman}.} determined by the braid as 
\be\label{hQ}
h^Q={\rm tr} \left( \prod_{i=1}^\infty \prod_{j=1}^{m-1}  R_j^{a_{ij}}\right).
\ee
These $R-$matrices (omitting henceforth the explicit reference to `quantum',  for simplicity) satisfy  a regular isotopic version of the HOMFLY--PT skein relation $R_j-R_j^{-1}=z\mathbb{I};\;\forall j\in\{1,\cdots,m-1\}$ and they can be determined using the representation theory of  the $SU_q(\lceil\frac{m}{2}\rceil)$ algebra \cite{Anokhina}. 
For instance, for $2,3$ and $4$ strands, the quantum group $SU_q(2)$ is enough for
their evaluation, but for 5 and 6 strands, one needs to invoke the $SU_q(3)$ algebra to find the associated Racah coefficients. Explicit expressions for up to $5$ strands  will be provided below.

The Schur functions\footnote{
The character $S_Q$ 
is related to the time variable $p_k=k t_k$ in KP $\tau$-function  
\cite{MorozovCH}, which is  chosen as $p_k^{*}= {\{A^k\}}/{\{q\}}$.\label{ft:qNUMBER}} $ S_Q$, which depend on both  the parameters $A=q^N$ and $q$, are expressed as
\be\label{S_Q^{hat}}
 S_Q=\prod_{(i,j)\in Q} \frac{\{A q^{i-j}\}}{\{q^{h_{i,j}}\}}.
\ee
Here $\{x\}:=x-x^{-1}$, while $(i,j)$ label the boxes of the Young diagram, with $i$ being the column index and $j$ the row index, starting from $(i,j)=(1,1)$ which corresponds to the box at the top-left corner. The denominator exponent  $h_{i,j}$ denotes the hook length, which is the number of  boxes to the right and below the box $(i,j)$, including itself.  In the classical limit, the quantum numbers $\{q^l\}$ are replaced by  ordinary numbers $l$, with which (\ref{S_Q^{hat}}) reduces to the  Weyl dimensional formula for the dimension $\delta_Q$  of the linear group representation corresponding to a Young diagram $Q$, as \cite{ItzyksonZuber}
\be\label{S_Qclassical}
 \delta_Q =\prod_{(i,j)}\frac{N+ i-j}{h_{i,j}}.
\ee
We also introduce the normalised version of the Schur functions
\be\label{S_Q^{*}}
S_Q^{*} = \frac{\{q\}}{\{A\}} S_Q,
\ee
which will be frequently used in the sequel since they are simpler to write down and can  directly  yield the normalised version of the  HOMFLY--PT polynomial as $H(\mathcal{K})=\frac{\{q\}}{\{A\}}\bar H=A^{-w}\sum_{Q} h^Q S_Q^*$.

Since the HZ transform (\ref{HZdef}) affects the variable $N$, it will be useful to concentrate all the $q^N=A-$dependence in the new function 
\be \label{Shat}
\hat S_Q:= A^{-w} S_Q.
\ee 
Hence, 
the HZ  transform of the HOMFLY--PT polynomial 
as given in (\ref{Racah1})  can be applied directly to $\hat S_Q$, yielding the HZ character expansion 
\be\label{prop2.2}
Z(\mathcal{K})= \sum_Q h^{Q}Z(\hat S_Q).
\ee
As we explain below 
this is a very useful formula for illuminating the factorisability properties of the  HZ function. To fully understand how this is achieved, we  unpack it explicitly  for 2,3,4 and 5 number of strands.

\paragraph{2 strands.}
There are two possible Young diagrams with two boxes: $ \img{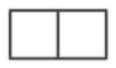}$ and $\img{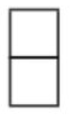}$. These will be often labelled as $[2]$ and $[11]$, respectively, in which the number of boxes contained in each row are indicated.
The corresponding (normalised) Schur functions    become   
\be S_{[2]}^{*}=S^*_{\img{2.png}} =\frac{\{ A q\}}{\{q^2\}},\;\;S_{[11]}^{*}=S^*_{\img{11.png}}= \frac{\{A q^{-1}\}}{\{q^2\}},\ee
while the Racah coefficients 
are $h^{[2]}=q^w$ and $h^{[11]}=(-q^{-1})^{w}$. In particular, $h^{[11]}$ is obtained from $h^{[2]}$ by $q\to -q^{-1}$, due to the `mirror' (row $\rightleftharpoons$ column)  symmetry that relates the diagrams $[2]$ and $[11]$.

\vskip2mm
\noindent{Example $3_1$.}
For the  2-strand braid $\sigma_1^{3}$, where  $3=w$, whose closure is the trefoil knot $3_1$, the HOMFLY--PT polynomial can be written as  
\ba \label{H31}
 H(3_1)\hspace{-2mm}&=&\hspace{-2mm} A^{-3}(q^3 S_2^{*}- q^{-3} S_{11}^{*})= A^{-3}\left(q^3 \frac{q A- q^{-1} A^{-1}}{q^2-q^{-2}}
- q^{-3}
\frac{q^{-1}A-q A^{-1}}{q^2-q^{-2}}\right)\nonumber\\
\hspace{-2mm}&=&\hspace{-2mm} A^{-2} (q^2+q^{-2}) - A^{-4}.
\ea
For
$A=a^{-1}$ (c.f. footnote~\ref{foot:A=a^-1}), this yields  the standard expression of the HOMFLY--PT polynomial for the right handed trefoil  \cite{KnotInfo}. 
More generally, the HOMFLY--PT for all 2-stranded torus links $T(2,n)$ ($n=1,2,3,...$) has the simple character expansion \cite{MorozovCHII}
\be
H(T(2,n)) = A^{-n}\left(q^n S_2^{*}+ \left(-q^{-1}\right)^n S_{11}^{*}\right).
\ee

The HZ transform  
of the characters $\hat S_Q=A^{-w}\frac{\{A\}}{\{q\}}S_Q^{*}$ in the 2-strand case
become
\ba
Z(\hat S_{2})= \frac{q^{-w} \lambda}{\mathcal{D}_2},\;\;Z(\hat S_{11})= \frac{q^{-2w}\lambda^2}{\mathcal{D}_2},
\ea
where $\mathcal{D}_2=(1-q^{-2-w}\lambda)(1-q^{-w} \lambda)(1-q^{2-w} \lambda)$.
Hence, the HZ function for a general 2-strand braid can be always written in the factorised form
\be\label{HZ2strand}
Z=\frac{\lambda}{\mathcal{D}_2} (1+(-1)^{w}q^{-3w}\lambda).
\ee
 In fact, the only links with braid index 2  are precisely the torus links $T(2,n)$, with $w=n$, whose HZ was found in \cite{Petrou,Petrou1} to agree\footnote{The HZ formulas presented here are in general related to the ones appearing in  \cite{Petrou,Petrou1} by $q\to q^{-1}$ (c.f. footnote~\ref{foot:A=a^-1}). 
   We further note that, if we denote  a Rolfsen knot with positive or negative writhe, respectively, by $\mathcal{K}^\pm$, then $Z(\mathcal{K}^\pm;\la,q^{-1})=Z(\mathcal{K}^\mp;\la,q)$. 
   } with (\ref{HZ2strand}). 

\paragraph{3 strands.}
The Schur functions in this case become 
\ba\label{S}
 &&S_3^{*}=S^*_{\img{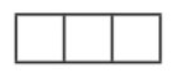}}=\{A q^2\}\{A q\}/\{q^3\}\{q^2\},
 \nonumber\\
 &&S_{111}^{*}=S^*_{\img{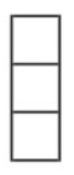}}= \{A q^{-2}\}\{A q^{-1}\}/\{q^3\}\{q^2\},\nonumber\\
&&S_{21}^{*}=S^*_{\img{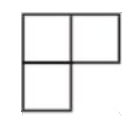}}=\{A q\}\{A q^{-1}\}/\{q^3\}\{q\}
\ea
and the Racah coefficients are $h^{[3]}=q^w$, $h^{[111]}=(-q^{-1})^{w}$, while for the hook-shape Young diagram  $[21]=\img{21elpizo.png}$ they are determined by the $R-$matrices 
\ba\label{eq:Racah3}
&&R_1^{[21]}= \begin{pmatrix} q & 0\\
0& -\frac{1}{q}\end{pmatrix}=:T,\;\;R_2^{[21]}=STS^{-1}= 
\begin{pmatrix} -\frac{1}{q^2}c & s\\
s& q^2 c\end{pmatrix}.
\ea 
Here $c= (q+q^{-1})^{-1}$, $s= \sqrt{q^2+1+ q^{-2}}/(q+q^{-1})$, satisfying $c^2+s^2=1$, and $S:= \begin{pmatrix} c & s\\
s& -c\end{pmatrix}$ is an orthogonal rotation matrix, known as a Racah matrix.

\vskip 2mm
\noindent{Example  $3_1$.}
The trefoil knot   can  also be expressed as a 3-strand braid $(\sigma_1\sigma_2)^2$ with $w=4$, 
 for which $h^{[21]}= {\rm tr}( R_1 R_2 R_1 R_2 )= -1$ and hence  the character expansion becomes
\be\label{T(3,2)}
 H(3_1) =A^{-4}( q^{4} S_3^{*}- S_{21}^{*}+ q^{-4 }S_{111}^{*})=A^{-4}(-1 + A^{2}(q^2+ q^{-2})),
\ee
which, as expected, agrees with (\ref{H31}).

\vskip 2mm
\noindent{Example  $4_1$.} The figure-8  knot can be written as the closure of the braid $(\sigma_1\sigma_2^{-1})^2$ with $w=0$ and hence $h^{[3]}=h^{[111]}=1$, while we compute 
\ba\label{t1}
h^{[21]}&=&\tr\left(R_1^{[21]}{R_2^{[21]}}^{-1}\right)^2={\rm tr} \left(\begin{pmatrix} q & 0\\
0& -\frac{1}{q}\end{pmatrix}
\begin{pmatrix} -\frac{1}{q^2}c & s\\
s& q^2 c\end{pmatrix}^{-1}
\right)^2\nonumber\\
&=&q^4 -2 q^2 + 1- 2 q^{-2}+ q^{-4}.
\ea
Using these coefficients and the characters in  $(\ref{S})$ we find the HOMFLY--PT polynomial
\be\label{H4_1}
 H
(4_1) = S_3^{*}+S_{111}^{*}+
(q^4-2 q^2 +1 - 2 q^{-2}+q^{-4})S_{21}^{*}=  A^2- q^2-q^{-2}+1 + A^{-2}.
\ee

The HZ transform of the unnormalised Schur functions $\hat S_{Q}=A^{-w}\frac{\{A\}}{\{q\}}S^*_Q$  reads
 \ba\label{ZC3_1}
 &&Z(\hat S_3)=
 \frac{q^{-w} \lambda}{\mathcal{D}_3},
 \hskip 2mm
 Z( \hat S_{21})=
 \frac{q^{-2w}(q^{-1}+q)\lambda^2}{\mathcal{D}_3},\;\; Z( \hat S_{111})=
 \frac{q^{-3w} \lambda^3}{\mathcal{D}_3},
 \ea
where the denominator is 
\be\label{D3w}
\mathcal{D}_3=(1-q^{-w-3} \lambda)(1-q^{-w-1} \lambda)(1-q^{-w+1} \lambda)(1-q^{-w+3} \lambda).
\ee 
These simple expressions are useful to understand the HZ character expansion  (\ref{prop2.2}) 
 for a general 3-strand braid, which can be written as 
\be\label{HZ3general}
 Z=\frac{\lambda}{\mathcal{D}_3}\left(1+\lambda h^{[21]}q^{-2w}(q+q^{-1})+(-1)^{w}q^{-4w}\lambda^2\right).
\ee

It is noteworthy  that the HZ functions for the characters in (\ref{ZC3_1}) have a factorised form 
and  that each Young diagram gives a contribution at different orders of $\lambda$.
The Jones
polynomial, corresponding to $N=2$, is obtained by the $\lambda^2$ term in the HZ transform. 
For small $\la$, using the Taylor expansion of (\ref{HZ3general}) with $q\to q^{-1}$, the  $\mathcal{O}(\lambda^2)$ term
yields the  unnormalised Jones polynomial $\bar{J}(q^2)=(q+q^{-1})J(q^2)= (q^{w-3}+q^{w-1}+q^{w+1}+q^{w+3})+ (q+q^{-1})q^{2w}h^{[21]}$, in which the first four terms come from  the expansion of the HZ denominator $\mathcal{D}_3$.  Hence, the Racah coefficient $h^{[21]}$ can be related to the Jones polynomial by
\be\label{h21a}
h^{[21]}=q^{-2w}\left( J(q^2) - q^{w-2}-q^{w+2}\right).
\ee 
This gives an alternative way to determine   $h^{[21]}$. 
For example, for $4_1$ the coefficient $h^{[21]}$  in (\ref{t1}) can be obtained via (\ref{h21a}) with $w=0$ and the Jones polynomial   $J(4_1;q^2)=q^4-q^2-q^{-2}+q^{
-4}+1$. 

\begin{prop}\label{prop:3fact}
    HZ factorisability for a knot with a  3-strand braid representative is admitted when 
the Racah coefficient $h^{[21]}$ is a  symmetric polynomial in $q$ with alternating coefficients $\pm1$, i.e. it is of the form
\be\label{h21fact}
h^{[21]}=-\frac{q^{\gamma}+q^{-\gamma}}{q+q^{-1}}=- q^{-\gamma-1}+q^{-\gamma+1}...\mp q^{-2}\pm 1\mp  ...- q^{\gamma+1};\;\;\gamma\in\mathbb{Z}_{\rm odd}.
\ee
Similarly, for 2-component links, factorisability occurs when 
\be\label{h21factLink}
h^{[21]}=-\frac{q^{\gamma}-q^{-\gamma}}{q+q^{-1}};\;\;\gamma\in \mathbb{Z}_{\rm even}.
\ee
When this is the case, the HZ function can be expressed in the factorised form 
\be\label{3factHZ}
Z=\frac{\lambda }{\mathcal{D}_3}(1-\lambda q^{-2w+\gamma})(1\pm\lambda q^{-2w-\gamma}).
\ee
\end{prop}
\begin{proof}
   This is easily verifiable from (\ref{HZ3general}), due to the multiplication of $h^{[21]}$ by $(q+q^{-1})$ and since  the writhe $w$ is even, respectively odd, for braid diagrams corresponding to   odd, respectively even, component links. 
     \end{proof}
As mentioned in the introduction, it was found in \cite{Petrou1} that concatenations of a braid with full twists  and Jucys--Murphy twists preserve HZ factorisability. Here we show that there is yet a third braid operation that preserves the factorisability condition above, which we call a \emph{partial full twist}. This is defined in the same way as a full twist, but it acts  only on  $m-1$ strands of an $m$-strand braid.   For instance, in the 3- and 4-strand cases the partial full twists $\sigma_1^2=:F_2\in B_3$ and  $(\sigma_2\sigma_1)^3=:F_3\in B_4$, respectively, are depicted in Fig.~\ref{fig:F2}. 
\begin{figure}[h]
    \centering
    \includegraphics[width=0.32\linewidth]{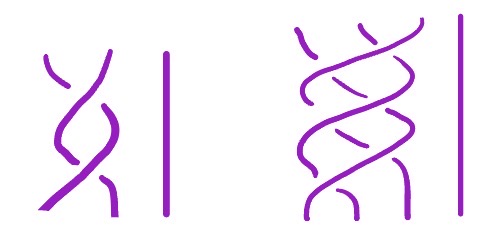}
    \caption{The partial full twists $F_2=\sigma_1^2\in B_3$ and $F_3=(\sigma_2\sigma_1)^3\in B_4$  involve only the first $m-1$ strands of an $m$-strand braid.}
    \label{fig:F2}
\end{figure}The preservation of the factorisability  condition in Prop.~\ref{prop:3fact} under these three braid operations 
can be understood by studying their $R-$matrix representations.
\begin{prop}\label{prop:F3E3}
    The matrix representation corresponding to a full twist $F_3=(\sigma_2\sigma_1)^{3}$ evaluates to be the $2\times 2$ identity matrix $\mathbb{I}_2$, i.e.
    \be
   \left(R_2^{[21]}
    R_1^{[21]}\right)^{3}=\mathbb{I}_2.
    \ee 
    For  partial full twists 
    $F_2=\sigma_1^{2}$ we compute
         \be 
{R_1^{[21]}}^{2}=\begin{pmatrix} q^{2}&0\\
0& q^{-2}\end{pmatrix},
    \ee
    while for  Jucys-Murphy twists 
 $\tilde{E}_3=(\sigma_2\sigma_1^2\sigma_2)$
     \be
     R_2^{[21]}{R_1^{[21]}}^2R_2^{[21]}=\begin{pmatrix} q^{-2}&0\\
0& q^{2}\end{pmatrix}.
    \ee
    The  representation corresponding to the combination $F_2^{ j}\otimes F_3^l\otimes\tilde{E}_3^k$ becomes
     \be
    \begin{pmatrix} q^{2j-2k}&0\\
0& q^{-2j+2k}\end{pmatrix}.
    \ee
    These three braid operations, which all commute with each other\footnote{ While for $F_3$
    any braid representation always yields the identity matrix,  for the equivalent form of the Jucys-Murphy twist  $E_3=\sigma_1\sigma_2^2\sigma_1$ and for the different choice of a partial full twist $\sigma_2^2$ (involving the last two, rather than the first two strands) the representations are not diagonal. However, the latter two still commute with each other.}, generate commutative subalgebras.
\end{prop}
\begin{proof}
    By direct computation using  (\ref{eq:Racah3}).
\end{proof}

The  HZ character expansion for the closure of full twists and Jucys--Murphy twists, which are 3-component links, are easily obtained by Prop.~\ref{prop:F3E3} to be
$ Z(F_3^l)=A^{-6l}\allowbreak(q^{6l} Z(\hat{S}_{3}) + 2 Z(\hat{S}_{21})+ q^{-6l} Z(\hat{S}_{111}))$ and 
$ Z(\tilde{E}_3^k)=A^{-4k}( (q^{2k}+q^{-2k})Z(\hat{S}_{21})\allowbreak +q^{4k} Z(\hat{S}_{3})  + q^{-4k} Z(\hat{S}_{111}))$, respectively.
\vskip2mm
\noindent{Remark.} It is important to explicitly point out the relations between the three braid operations  that generate commutative subalgebras. 
    For instance, full twists 
    can be constructed by combinations of Jucys-Murphy and partial full twists in two different forms, as\footnote{This can be easily observed in the 3- and 4-strand cases from Table~1 in \cite{Petrou1}, 
    in which knots in the vertical direction are connected by successive concatenation with partial full twists.}
    \be
    F_m=(\sigma_{m-2}\cdots\sigma_1)^{m-1}\otimes\tilde{E}_m=(\sigma_{m-1}\cdots\sigma_2)^{m-1}\otimes E_m.
    \ee
    Alternatively,  full twists can be constructed as products of Jucys-Murphy twists applied on decreasing number of strands, as 
\be
F_m=\tilde{E}_m\otimes \tilde{E}_{m-1}\otimes\cdots\otimes \tilde{E}_2=E_m\otimes E_{m-1}\otimes\cdots\otimes E_2.
\ee
Here, for $i<m$, $\tilde{E}_i$ involves the first $i$ strands while $E_i$ involves the last $i$ strands.
For instance, for 3 strands this reads  $F_3=\tilde E_3\otimes \tilde E_2:=(\sigma_2\sigma_1^2\sigma_2)\otimes(\sigma_1^2)$ or $F_3=E_3\otimes E_2:=(\sigma_1\sigma_2^2\sigma_1)\otimes(\sigma_2^2)$. 
There is also a third relation, obtained by defining yet another, equivalent version of the 3-strand Jucys-Murphy braid  $\tilde{\tilde{E}}_3:=\sigma_2^2\sigma_1^2$, with which 2 copies of the 3-strand full twist can be constructed as  $F_3^2=E_3\otimes \tilde{\tilde{E}}_3\otimes \tilde{E}_3$.  This also applies to higher number of strands, with a product of $m$ different versions of the Jucys-Murphy braid yielding $F_m^2$.

\vskip2mm
Prop.~\ref{prop:F3E3} can be used to shown that arbitrary concatenations (in any order) of $F_2$, $F_3$ and $\tilde{E}_3$ to a braid $\mathbf{b}$ that satisfies Prop.~\ref{prop:3fact}, preserve HZ-factorisability. We shall  refer to such a braid as the \emph{base braid} or, by slight abuse of language, the \emph{base knot} (thought of as the closure of the braid).

\begin{cor}
Concatenation by $F_2^{ j}\otimes \tilde{E}_3^k\otimes F_3^l $ adds  $2j+6l+4k$ (the number of the added crossings) to its writhe, hence mapping
\be
h^{[3]}(\mathbf{b})\to q^{2j+6l+4k}h^{[3]}(\mathbf{b}).
\ee
By Prop.~\ref{prop:F3E3}, these also leave the form of $h^{[21]}$ unaffected. 
 In particular, for a base knot  $\mathbf{b}$ which has an $R-$matrix representation with entries $\{x_{ij}\}$, we find
    \be
    h^{[21]}(\mathbf{b}\otimes F_2^{ j}\otimes F_3^l\otimes \tilde{E}_3^k)=x_{11}q^{ 2j-2k}+x_{22}q^{- 2j+2k}.
    \ee
    Hence if  $h^{[21]}(\mathbf{b})=x_{11}+x_{22}$ satisfies the  condition of Prop.~\ref{prop:3fact}, so does $h^{[21]}(\mathbf{b}\otimes F_2^j \otimes\tilde{E}_3^k \otimes F_3^l)$.
\end{cor}

This can be used to prove the HZ-factorisability for a general family of 3-strand knots generated by these braid operations. Such a family has base braid $\mathbf{b}=\sigma_2\sigma_1$, whose closure is the unknot, and 
its matrix representation has diagonal entries $x_{11}^{[21]}=\frac{-q^{-1}}{q+q^{-1}},$ $x_{22}^{[21]}=\frac{-q}{q+q^{-1}}$. 
Hence
\be
h^{[21]}(\sigma_2\sigma_1\otimes F_2^{ j}\otimes F_3^l\otimes \tilde{E}_3^k)=-\frac{q^{- 2j +2k+1}+q^{ 2j-2k-1}}{q+q^{-1}},
\ee
which together with the writhe $w=2+2j+6l+4k$, provides all the information needed to obtain the factorised HZ function (\ref{3factHZ}).


We shall henceforth collectively denote this general HZ-factorisable family  by \be\boxed{\mathcal{K}^{(3)}_{j,k,l}:=\sigma_2\sigma_1\otimes F_2^{ j}\otimes F_3^l\otimes \tilde{E}_3^k,}\ee 
a member of which is depicted in Fig.~\ref{fig:K111}.
This is a \emph{hyperbolic extension of $3$-strand torus knots}, since the latter, which are mainly constructed by full twists, are included as the special cases $\mathcal{K}^{(3)}_{0,0,l}=\mathcal{K}^{(3)}_{s,s,l-s}=T(3,3l+1)$ and $\mathcal{K}^{(3)}_{1,0,l-1}=\mathcal{K}^{(3)}_{l-s,l-s-1,s}=T(3,3l-1)$. 
 The inclusion of arbitrary numbers of partial full twists and Jucys--Murphy twists introduces hyperbolicity.
\begin{figure}[h]
    \centering
    \includegraphics[width=0.47\linewidth]{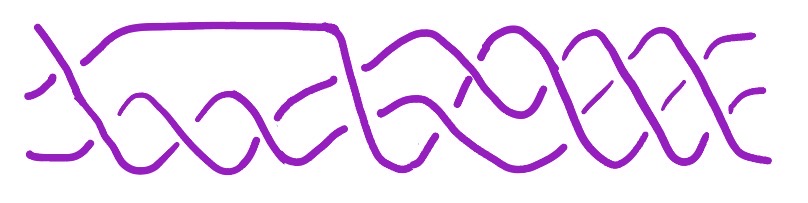}
    \caption{The knot $\mathcal{K}^{(3)}_{1,1,1} $ is the closure of this braid, containing one partial full twist $F_2$, one Jucys-Murphy twist $E_3$ and one full twist $F_3$.  
   \vspace{2mm} }
    \label{fig:K111}
\end{figure}
 \vskip2mm
\noindent{Remark.} All the 3-strand HZ-factorisable knots found in \cite{Petrou1}
are members of the family ${K}^{(3)}_{j,k,l}$ 
(which is equivalent with $\sigma_1\sigma_2^{\pm(1+2j)}\otimes F_3^l\otimes E_3^k$, with $j>0$), for some special values of $j,k,l\in\mathbb{Z}$.   This is clearly indicated for knots with up to 13 crossings in Table~1 of \cite{Petrou1}, while other cases include $\mathcal{K}_{j,k}=\mathcal{K}^{(3)}_{j,k,0}$, 
$T(3,2l+1,2,2j)=\sigma_2\sigma_1\otimes F_2^{j}\otimes F_3^l={K}^{(3)}_{j,0,l}$, 
$T(3,3l+1)\otimes \tilde{E}_3^k=\sigma_2\sigma_1\otimes \tilde{E}_3^k\otimes F_3^l={K}^{(3)}_{0,k,l}$  and $T(3,3l+2)\otimes \tilde{E}_3^k=\sigma_2\sigma_1\otimes F_2\otimes \tilde{E}_3^k\otimes F_3^l={K}^{(3)}_{1,k,l}$. 
\vskip2mm

Similarly we can define the most general HZ-factorisabily family of 3-strand hyperbolic links as
\be\boxed{\mathcal{L}^{(3)}_{j,k,l}:=\sigma_2\otimes F_2^{ j}\otimes \tilde{E}_3^k\otimes F_3^l,}\ee 
for which 
\be\label{3linkstrandGenRacah}
h^{[21]}(\mathcal{L}^{(3)}_{j,k,l})=-\frac{q^{-2 + 2 j - 2 k}-q^{2 - 2 j + 2 k}}{q+q^{-1}}.
\ee
As for knots, this family contains all the 3-strand factorised cases presented in \cite{Petrou1}, since it can equivalently be expressed as  $\sigma_1\otimes \sigma_2^{2j}\otimes E_3^k\otimes F_3^l$ (c.f. Table~2 in \cite{Petrou1}). Explicitly, $P(-2,3,2j) 
=\mathcal{L}^{(3)}_{j-1,0,1}=\mathcal{L}^{(3)}_{j,1,0}$, $\mathcal{L}_{j,k}=\mathcal{L}^{(3)}_{j,k,0}$, $L9n14\{0\}\otimes E_3^k=\mathcal{L}^{(3)}_{-2,k+1,0}$, $T(3,3l,2,1)\otimes\tilde{E}_3^k=\mathcal{L}^{(3)}_{1,k+1,l}$, 
$T(3,3l+1,2,1)\otimes\tilde{E}_3^k=\mathcal{L}^{(3)}_{2,k+1,l}$, 
$T(3,3l+2,2,1)\otimes\tilde{E}_3^k=\mathcal{L}^{(3)}_{3,k+1,l}$. 


\paragraph{4 strands.}
There are 5  characters $S_Q^{*}$ in this case, corresponding to the  Young diagrams  $[4],[31],[22],[211],[1111]$. By (\ref{S_Q^{*}}), these read 
\ba\label{S_4}
&&S_{[4]}^{*}= 
\frac{\{A q^3\}\{A q^2\}\{A q\}}{\{q^4\}\{q^3\}\{q^2\}}, \hskip 2mm
S_{[31]}^{*}=\frac{\{A q^2\}\{A q\}\{A q^{-1}\}}{\{q^4\}\{q^2\}\{q\}}, \nonumber\\
&& S_{[22]}^{*}= \frac{\{A q\}\{A q^{-1}\}\{A \}}{\{q^3\}\{q^2\}\{q^2\}},\;\;S_{211}^{*}= \frac{\{A q^{-2}\}\{A q^{-1}\}\{A q\}}{\{q^4\}\{q^2\}\{q\}},\nonumber\\
&&\hspace{15mm}
S_{[1111]}^{*}=\frac{\{A q^{-3}\}\{A q^{-2}\}\{A q^{-1}\}}{\{q^4\}\{q^3\}\{q^2\}}.
\ea
The $q$-characters $S_Q=\{A\}/\{q\}S^*_Q$ reduce to the corresponding  ordinary numbers 
$
\delta_{[4]} = \frac{N}{24}{(N+3)(N+2)(N+1)},$ $
\delta_{[31]}=\frac{N}{8}
(N+2)(N+1)(N-1),
\delta_{[22]} = \frac{N}{12}N (N+1)(N-1),
\delta_{[211]} = \frac{N}{8} (N-2)(N-1) (N+1),$ and $
\delta_{[1111]} = \frac{N}{24}(N-3)(N-2)(N-1)$,
which agree with \cite{ItzyksonZuber}.

The Racah coefficients $h^{[31]}$ for 4-strand braids can be computed using the  $R-$matrices \cite{Anokhina}
\be\label{4-strandCH}
 R_1^{[31]}= \begin{pmatrix} q & &\\
& q&\\
& &-\frac{1}{q}\end{pmatrix}, R_2^{[31]}= U R_1^{[31]} U^{-1},\;\; R_3^{[31]}= U V R_1^{[31]} (UV)^{-1},
\ee
in which $U,V$ are the Racah matrices
 \be\label{R3}
 U= \begin{pmatrix} 1 & 0&0\\
0& c& s\\
0& -s&c\end{pmatrix},\;\;
 V= \begin{pmatrix} \frac{1}{[3]_q} & \frac{[2]_q \sqrt{q^2+q^{-2}}}{[3]_q}&0\\
-\frac{[2]_q \sqrt{q^2+q^{-2}}}{[3]_q}& \frac{1}{[3]_q}&0\\
0& 0&1\end{pmatrix},
\ee
where
$[2]_q= q+q^{-1},\;\; [3]_q = q^2+1+q^{-2}$ and $s,c$ are defined below (\ref{eq:Racah3}).
Explicitly, the $R-$matrices for $[31]$ can be written as
\be\label{R2R3[31]}\nonumber
R_2^{[31]}= \begin{pmatrix} q & 0&0\\
0& \frac{q}{(q^{-1} + q)^{2}} - \frac{1 + 1/q^2 + q^2}{q (q^{-1} + q)^{2}}& -\frac{\sqrt{1 + 1/q^2 + q^2}}{q (q^{-1} + q)^2)} - \frac{
 q \sqrt{1 + q^{-2} + q^2}}{(1/q + q)^2}\\
0& -\frac{\sqrt{1 + 1/q^2 + q^2}}{q (1/q + q)^2} - \frac{
 q \sqrt{1 + 1/q^2 + q^2}}{(1/q + q)^2}
&-\frac{1}{q (1/q + q)^2} + \frac{q (1 + 1/q^2 + q^2)}{(1/q + q)^2}
\end{pmatrix},
\ee
\be\label{R2[31]}
R_3^{[31]}=\begin{pmatrix} - \frac{1}{q(1-q+q^2)(1+q+q^2)} & - \frac{q(1+q^2)\sqrt{q^{-2}+q^2}}{(1-q+q^2)(1+q+q^2)}& 0\\
- \frac{q(1+q^2)\sqrt{q^{-2}+q^2}}{(1-q+q^2)(1+q+q^2)}& \frac{q^5}{(1-q+q^2)(1+q+q^2)} &0\\
0&0&q
\end{pmatrix}.
\ee
These satisfy 
(for $n\ge 1$)
\ba
{\rm det}R_1^{[31]}&=&{\rm det}R_2^{[31]}={\rm det}R_3^{[31]}= -q,\nonumber\\
{\rm tr}\big(R_i^{[31]}\big)^n&=&
2 q^n + (-q^{-1})^n;\;\;i=1,2,3\nonumber\\
{\rm tr}\left(R_1^{[31]}\right)\cdot {\rm tr}\left((R_2^{[31]})^{-1}\right)-1&=& {\rm tr}\left(R_1^{[31]}\cdot (R_2^{[31]})^{-1}\right)+ {\rm tr}\left((R_1^{[31]})^{-1}\cdot R_2^{[31]}\right)\nonumber\\
&=& 4- \frac{2}{q^2}- 2 q^2.
\ea
 The $R-$matrices corresponding to the Young diagram $Q=[22]$ are
 \be\nonumber
 R_1^{[22]} = \begin{pmatrix} q & 0\\
0& -\frac{1}{q}\end{pmatrix},\hskip 2mm  R_3^{[22]}=R_1^{[22]},
\ee
\be\label{R22a}
R_2^{[22]}=U^{[22]}R_1^{[22]}(U^{[22]})^{-1} =
\begin{pmatrix} -c q^{-2} & -s\\
-s& q^2 c\end{pmatrix}=\begin{pmatrix} -\frac{1}{q^2[2]_q} & -\frac{\sqrt{[3]_q}}{[2]_q}\\
  -\frac{\sqrt{[3]_q}}{[2]_q}& \frac{q^2}{[2]_q}\\
\end{pmatrix}.
\ee
The Racah coefficients for  a 4-strand braid $\prod_i\sigma_1^{a_i}\sigma_2^{b_i}\sigma_3^{c_i}$ with writhe $w=\sum_i (a_i + b_i +c_i)$, are expressed in terms of the above matrices as
\be
h^Q={\rm tr} \bigg[ \prod_i(R_1^Q)^{a_i} (R_2^Q)^{b_i} (R_3^Q)^{c_i} \bigg].\ee
Since the Young diagram $[211]$ is the `mirror' of $[31]$, the  coefficient $h^{[211]}$ is obtained from $h^{[31]}$  by the replacement $q\to - q^{-1}$. The same holds for $[4]$ and $[1111]$, for which $h^{[4]}= q^{w}$ and $h^{[1111]}= (- q^{-1})^w$, as before.
\vskip 2mm
\noindent{Example $6_1$.}
A braid for $6_1$ is $\sigma_1^{-1}\sigma_2\sigma_3^{-1}\sigma_1^{-1}\sigma_2\sigma_3^2$. 
The coefficient $h^{[22]}$, using (\ref{R22a}),  becomes $h^{[22]}=q-q^{-1}$, 
while $h^{[31]}$ is evaluated from (\ref{4-strandCH}) and (\ref{R2[31]}) as $
 h^{[31]}=q^{-5}-q^{-3}-q^{-1}+q- 2 q^{3}+q^{5}.$ 
 The HOMFLY--PT polynomial of $6_1$ is obtained via the character expansion 
\ba\label{CH6_1}
 H(6_1)&=&A^{-1}( q S_4^{*}+ (q^{-5}-q^{-3}-q^{-1}+q- 2 q^{3}+q^{5}) S_{31}^{*}
 +(q-q^{-1})S_{22}^{*}\nonumber\\
 &&+ (-q^5+  q^3+q-q^{-1}+2 q^{-3}-q^{-5}) S_{211}^{*}-q^{-1}S_{1111}^{*})\nonumber\\
 &=&A^{-4}+A^{2}+(1 - q^{-2} - q^2)A^{-2}-  q^{-2}(-1 + q^2)^2,
 \ea 
 in agreement with \cite{KnotInfo} (with $A=a^{-1}$). 
The HZ transform of the characters $\hat S_Q$, takes the form
\ba\label{4-strandHZ}
&&Z(\hat S_{4})= \frac{q^{-w} \lambda}{\mathcal{D}_4},\hskip 2mm
Z(\hat S_{31})= \frac{q^{-2w}\lambda^2}{\mathcal{D}_4} (q^{-2}+1+q^2),
\nonumber\\
&&\hspace{12mm}Z(\hat S_{211})=  \frac{q^{-3w}\lambda^3}{\mathcal{D}_4}  (q^{-2}+1+q^2), \nonumber\\
&&\hspace{-5mm}Z(\hat S_{22})= \frac{1}{\mathcal{D}_4}q^{-2w} \lambda^2 (1+ q^{-w} \lambda),\;\;Z(\hat S_{1111})= \frac{q^{-4w} \lambda^4}{\mathcal{D}_4},
\ea
where the denominator $\mathcal{D}_4$, for arbitrary $w$, reads
\be\label{D4w}
\mathcal{D}_4 = (1-q^{-w-4}\lambda)(1-q^{-w-2}\lambda)(1-q^{-w} \lambda)(1-q^{-w+2}
\lambda)(1-q^{-w+4} \lambda).
\ee
Hence the HZ character expansion for  4-strand braids becomes
\ba\label{Z4}\nonumber
 Z&=&\frac{\lambda }{\mathcal{D}_4}\bigg(1+\lambda q^{-2w} h^{[31]}(q^2+1+q^{-2})+\lambda q^{-2w} h^{[22]}(1+q^{-w}\lambda)\\&&\hspace{9mm}+\lambda^2 q^{-3w}h^{[211]}
 (q^2+1+q^{-2})
 +(-1)^{w}q^{-5w}\lambda^3\bigg).
\ea
\begin{prop}\label{prop:4fact}
Factorisability of the HZ function of a 4-strand knot  occurs  when
\be\label{4condition}
h^{[22]}=0,\;\;h^{[31]}=-\frac{q^{\gamma_1}+q^{\gamma_2}+q^{\gamma_3}}{q^2+1+q^{-2}};\;\;\gamma_i\in\mathbb{Z}_{\rm odd}\;{\rm s.t.}\;\sum_{i}\gamma_i=w.
\ee
Similarly, for 2-component links the conditions become 
 \be\label{4conditionlinks}
h^{[22]}=0,\;\;h^{[31]}=-\frac{-q^{\gamma_1}+q^{\gamma_2}+q^{\gamma_3}}{q^2+1+q^{-2}};\;\;\gamma_i\in\mathbb{Z}_{\rm even}\;{\rm s.t.}\;\sum_{i}\gamma_i=w.
\ee
These yield the factorised HZ function 
\be\label{Z4fact}
 Z=\frac{\lambda} {\mathcal{D}_4}(1\pm\lambda q^{-2w+\gamma_1})(1-\lambda q^{-2w+\gamma_2})(1-\lambda q^{-2w+\gamma_3}).
 \ee
 \end{prop}
\begin{proof}
After  multiplying $h^{[31]}$ in (\ref{4condition}) with $q^2+1+q^{-2}$, which appears in $Z(\hat{S}_{31})$, it simplifies to $-(q^{\gamma_1}+q^{\gamma_2}+q^{\gamma_3})$. Similarly, in the case of $[211]$ it becomes $q^{-\gamma_1}+q^{-\gamma_2}+q^{-\gamma_3}$. By substituting these together with $h^{[22]}=0$ in (\ref{Z4}) and noting that $w$ is odd, resp. even, in the case of a link with  odd, resp. even, number of components,  (\ref{Z4fact}) is obtained.
\end{proof}

\begin{prop}\label{prop:F4E4}
    The matrix representations corresponding to a full twist $F_4=(\sigma_3\sigma_2\sigma_1)^{4}$ is expressed in terms of the $n\times n$ identity matrices $\mathbb{I}_n$ as
    \be
    \left(R_3^{[31]}R_2^{[31]}
    R_1^{[31]}\right)^{4}=q^{4}\mathbb{I}_3,\;\; \left(R_3^{[22]}R_2^{[22]}
    R_1^{[22]}\right)^{4}=\mathbb{I}_2.
    \ee
    For a partial full twist $F_3=(\sigma_2\sigma_1)^{3}$ we compute
       \be
\left(R_2^{[31]}R_1^{[31]}\right)^{3}=\begin{pmatrix} q ^{6}& \\
& \mathbb{I}_2\end{pmatrix},\;\; \left(R_2^{[22]}
    R_1^{[22]}\right)^{3}=\mathbb{I}_2.
    \ee
    For the Jucys-Murphy twist\footnote{For the different version of the Jucys-Murphy twist $E_4=\sigma_1\sigma_2\sigma_3^2\sigma_2\sigma_1$ it still holds that for $Q=[22]$ the product of $R-$matrices is $\mathbb{I}_2$, however for $Q=[31]$ it no longer diagonal. The same holds for a different choice of partial full twist $(\sigma_3\sigma_2)^3$.}  $\tilde{E_4}=\sigma_3\sigma_2\sigma_1^2\sigma_2\sigma_3$, the matrix representation obtained by the product $R_3^{Q}R_2^{Q}{R_1^{Q}}^2R_2^{Q}R_3^{Q}$,  for the following two choices of $Q$ becomes
    \be
Q=[31]: \begin{pmatrix} q ^{-2}& \\
& q^{4}\mathbb{I}_2\end{pmatrix},\;\; Q=[22]:\mathbb{I}_2.
    \ee 
    Hence, the combination $F_3^{ j}\otimes F_4^l\otimes \tilde{E}_4^k$ for $j,k,l\in\mathbb{Z}$ yields
\be
Q=[31]: \begin{pmatrix} q^{6j+4l-2k}& \\
& q^{4(k+l)}\mathbb{I}_2\end{pmatrix},\;\;Q=[22]: \mathbb{I}_2.
    \ee
  As in the 3-strand case, these braid operation commute with each other and generate commutative subalgebras.
\end{prop}
\begin{proof}
    By direct computation using and (\ref{4-strandCH}) and (\ref{R2R3[31]}).
\end{proof}

By Prop.~\ref{prop:F4E4}, the closure of full twists $F_4^l$ and Jucys--Murphy braids $E_4^k$ or $\tilde{E}_4^k$ 
 have HZ expansion
$ Z(F_4^l)=A^{-12l}(q^{12l} Z(\hat{S}_{[4]})+3q^{4l}Z(\hat{S}_{[31]}) + 2 Z(\hat{S}_{[22]}) +3q^{-4l}Z(\hat{S}_{[211]})\allowbreak + q^{-12l} Z(\hat{S}_{[1111]}))$ and $Z(E_4^k)=A^{-6k}(q^{6k} Z(\hat{S}_{[4]}) +2Z(\hat{S}_{[22]})+ (2q^{4k}+q^{-2k}) Z(\hat{S}_{[31]}) \allowbreak +(2q^{-4k}+q^{2k})Z(\hat{S}_{[211]})+ q^{-6k} Z(\hat{S}_{[111]}))$, respectively.

 \begin{cor}
Concatenations with $ F_3^j\otimes F_4^l\otimes \tilde{E}_4^k$ to a base braid $\mathbf{b}$
 changes the Racah coefficients as
\be
h^{[4]}(\mathbf{b})\to q^{12l+6j+6k}h^{[4]}(\mathbf{b}), 
\;\;h^{[22]}(\mathbf{b})\to h^{[22]}(\mathbf{b}).
\ee 
 For $Q=[31]$, if the $R-$matrix representation for $\mathbf{b}$  has entries $x^{[31]}_{ij}$, then the corresponding Racah coefficient becomes 
    \be
    h^{[31]}(\mathbf{b}\otimes F_3^j\otimes F_4^l\otimes \tilde{E}_4^k)=x_{11}^{[31]}q^{6j+4l-2k}+x_{22}^{[31]}q^{4(k+l)}+x_{33}^{[31]}q^{4(k+l)}.
    \ee
    Hence if $h^{Q}(\mathbf{b})=\sum_ix_{ii}^{Q}$ satisfy the HZ factorisability conditions of Prop.~\ref{prop:4fact} then so do $h^{Q}(\mathbf{b}\otimes F_3^j\otimes F_4^l\otimes \tilde{E}_4^k)$.
    \end{cor}
    \vskip2mm
\noindent{\bf $\bullet$ Base braid $\mathbf{b}=\sigma_3\sigma_2\sigma_1$.} For this braid,
 whose closure is the unknot, we compute $x_{11}^{[31]}=\frac{-q^{-1}}{[3]_q},$ $x_{22}^{[31]}=\frac{-q^2}{[3]_q[2]_q}$ and $x_{33}^{[31]}=\frac{-q^2}{[2]_q}$ 
(recall $[3]_q=q^2+1+q^{-2}$, $[2]_q=q+q^{-1}$)
and hence for the hyperbolic extension of $4$-strand torus knots \be\mathcal{K}_{j,k,l}^{(4)}=\sigma_3\sigma_2\sigma_1\otimes F_3^j\otimes F_4^l\otimes \tilde{E}_4^k\ee
we find
\be\label{K431}
h^{[31]}(\mathcal{K}_{j,k,l}^{(4)})=-\frac{q^{6j-2k+4l-1}+q^{4(l+k)+1}+q^{4(l+k)+3}}{q^2+1+q^{-2}}.
\ee
  The sum of the numerator exponents in (\ref{K431}) is equal to the writhe $w=3+6j+12l+6k$.
Similarly we compute $x_{11}^{[22]}=-1/[2]_q$, $x_{22}^{[22]}=1/[2]_q$, so $h^{[22]}=0$ for all members of the family,  and hence it satisfies the factorisability conditions in Prop.~\ref{prop:4fact}. Indeed this family contains torus knots  $T(4,4l+1)=\mathcal{K}_{0,0,l}^{(4)}=\mathcal{K}_{l-s,l-s,s}^{(4)}$ (when $k=j$) and $T(4,4l-1)=\mathcal{K}_{1,0,l-1}^{(4)}=\mathcal{K}_{l-s,l-s-1,s}^{(4)}$  (when $k=j-1$), and further includes $10_{132}=\mathcal{K}_{-2,0,1}^{(4)}$, 
$T(4,4l+1)\otimes E_4^k=\mathcal{K}_{0,k,l}^{(4)}$ and $T(4,4l+1,3,3j)=\mathcal{K}_{j,0,l}^{(4)}$, which were found to be factorised in \cite{Petrou1}.

\vskip 2mm
\noindent{Example  $10_{132}=\mathcal{K}^{(4)}_{-2,0,1}$
.} 
Using the braid   $\sigma_1^3\sigma_2^{-1}\sigma_1^{-2}\sigma_2^{-1}\sigma_3^{-1}\sigma_2\sigma_3^{-2}$, with $w=-3$,  we find 
$h^{[22]}= 0$ 
and,   similarly,
 \be\label{h31:10_132}
 h^{[31]}=-q^{-5}+q^{-1}-q+q^5 -q^7=-\frac{q^{-7} +q^{-5} + q^9}{q^2+1+q^{-2}}.
 \ee
  Substituting in (\ref{Z4}), this leads to the factorised HZ function  $(1-\lambda q^{15})(1-\lambda q)(1-\lambda q^{-1})/(1-\lambda q^{7})(1-\lambda q^{5})(1-\lambda q^{3})(1-\lambda q)(1-\lambda q^{-1})$, in which the numerator exponents are indeed given by $\gamma_i-2w$, in agreement with Prop.~\ref{prop:4fact}.

 \vskip 2mm
\noindent{Example $T(4,n)$.} 
For  4-strand torus knots and links $T(4,n)$, the Racah coefficients are obtained by $h^{Q}=\tr((R_3^{Q}R_2^{Q}R_1^{Q})^n)$.
When $n=4l\pm1$ is odd (corresponding to torus knots), $
h^{[31]}=- q^{n}$and $h^{[22]}=0$;
while when $n=6,10,14,...$, corresponding to torus links with 2 components (for which the HZ transform is not factorisable),
$
h^{[31]}=- q^{n} $and $h^{[22]}=2$.
When $n=4k$, 
$T(4,4k)$ correspond to full twists $F_4^k$, for which $h^{[31]}=
3 q^{n}$ and $h^{[22]}=0$, as mentioned above.

\vskip2mm
\noindent{\bf $\bullet$ Base braid $\mathbf{b}=\sigma_3\sigma_1^{-1}\sigma_2^{-2}\sigma_1\sigma_2^{-1}\sigma_1^{-1}$.}
For this braid, corresponding to a 4-strand version of $5_2^-$, we find $x_{11}^{[31]}=\frac{-q^{-7}}{[3]_q},$ $x_{22}^{[31]}=\frac{-q^2(q^{-4}-q^{-2}+1-q^2+q^4)}{[3]_q[2]_q}$ and $x_{33}^{[31]}=\frac{-q^2(q^{-4}-q^{-2}+1-q^2+q^4)}{[2]_q}$, yielding
\be
h^{[31]}(\mathbf{b}\otimes F_3^{ j}\otimes F_4^l\otimes \tilde{E}_4^k)=-\frac{q^{ 6j-2k+4l-7}+q^{4(l+k)-3}+q^{4(l+k)+7}}{q^2+1+q^{-2}}
\ee
and, similarly, $h^{[22]}=0$. Hence, by Prop.~\ref{prop:4fact} this family is also HZ-factorisable. The knot $10_{128}$ corresponds to $j=0,k=0,l=1$, while $12n_{318}$ corresponds to $j=1,k=0,l=1$. Note that the alternative braid representation of $5_2^-$, $\mathbf{b}=\sigma_3\sigma_2^{-3}\sigma_1^{-1}\sigma_2\sigma_1^{-1}$, has the same diagonal elements $x_{ii}^{Q}$ and hence yields the same results. However, this does not hold for the  braid  representative $\sigma_3\sigma_1^{-1}\sigma_2\sigma_1^{-1}\sigma_2^{-2}\sigma_1^{-1}$, for which $x_{ii}^{Q}$ do not satisfy the conditions in Prop.~\ref{prop:4fact}.

\vskip2mm
\noindent{Remark.}
4-strand versions of the knots $8_{20}$, $10_{125}$, etc., which are related to $5_2$ by additional 2-strand full twists $F_2$ can not serve as base knots to yield HZ factorisable families. This can be understood by the fact that, while ${R_1^{[31]}}^2$ 
is diagonal,  
\be
{R^{[22]}}^2=\begin{pmatrix} q^2 & \\
& q^{-2}\end{pmatrix},
\ee
which is not the identity matrix and hence it does not preserve the condition $h^{[22]}=0$. 
\vskip2mm
It is important to clarify that the 
 conditions in Prop.~\ref{prop:4fact} are \emph{sufficient} for HZ factorisability but they are not necessary. This is indicated by the example of the HZ-factorisable torus knot $T(3,5)=10_{124}$ which can also be expressed as the closure of the braid $\sigma_3\sigma_2^{-1}\sigma_1^{-1}\otimes F_3^{-1}\otimes F_4\otimes E_4$, 
 for which $h^{[31]}=-q+q^3-q^5$ and $h^{[22]}=q^{-1}-q$. Another interesting example is the case of the twisted torus links $T(4,4l+1,2,1)$, which are the closure of $\sigma_3(\sigma_3\sigma_2\sigma_1)^{4l+1}$, and for which the coefficients become $h^{[31]}=0$ and $h^{[22]}=-1$. These still lead to the factorised form (\ref{Z4fact}) with $\gamma_1=\gamma_2=w/2$ and $\gamma_3=-w$. 
Moreover, the 
 link $L10n_{42}\{1\}$, which was found in \cite{Petrou1} to be a special case  admitting HZ factorisability but not satisfying the HOMFLY--PT/Kauffman relation, is also an exceptional case. With  braid representative $\sigma_1\sigma_2^{-1}\sigma_1\sigma_3^{-1}\sigma_2^2\sigma_3^{-1}\sigma_2^{-1}\sigma_1\sigma_3$ we compute the Racah coefficients 
$h^{[31]}=-q^{-4}+2q^{-2}-3+2q^2-q^6+q^8$ and  $h^{[22]}=-q^{-4}+q^{-2}-1+q^2-q^4$, which clearly do not satisfy the conditions in Prop.~\ref{prop:4fact}, 
but due to a simplification still lead to HZ factorisation.  It is remarkable that the  $L10n_{42}\{1\}$ is the only link among the HZ-factorisable cases known to us that can not be obtained  as the concatenation with $F_3$, $F_4$ or $\tilde{E}_4$ to a  base link that is HZ-factorisable. In fact, none of the  2-component 4-strand  links that admit HZ-factorisability we are aware of 
have Racah coefficients  of the form (\ref{4conditionlinks}). 


\paragraph{5 strands.}
The Schur functions for the 5-strand case are 
\ba\label{S_5}
&&S_{5}^{*}= \frac{\{A q^4\}\{A q^3\}\{A q^2\}\{Aq\}}{\{q^5\}\{q^4\}\{q^3\}\{q^2\}},\hskip 2mm
S_{41}^{*}=\frac{\{A q^3\}\{A q^2\}\{A q\}\{A q^{-1}\}}{\{q^5\}\{q^3\}\{q^2\}\{q\}},\nonumber\\
&&S_{32}^{*}=\frac{\{A q^2\}\{A q\}\{A \}\{A q^{-1}\}}{\{q^4\}\{q^3\}\{q^2\}\{q\}},\hskip 2mm
S_{311}^{*}=\frac{\{A q^2\}\{A q\}\{A q^{-1} \}\{A q^{-2}\}}{\{q^5\}\{q^2\}\{q^2\}\{q\}},\nonumber\\
&&\hspace{-3mm}S_{221}^{*}=\frac{\{A q^{-2}\}\{A q^{-1}\}\{A  \}\{A q\}}{\{q^4\}\{q^3\}\{q^2\}\{q\}},\hskip 2mm
S_{2111}^{*}=\frac{\{A q^{-3}\}\{A q^{-2}\}\{A  q^{-1}\}\{A q\}}{\{q^5\}\{q^3\}\{q^2\}\{q\}},\nonumber\ea
\ba
&&\hspace{1mm}S_{11111}^{*}=\frac{\{A q^{-4}\}\{A q^{-3}\}\{A q^{-2}\}\{A q^{-1}\}}{\{q^5\}\{q^4\}\{q^3\}\{q^2\}}.
\ea
Their HZ transform, after the multiplication of a factor $A^{-w}\{A\}/\{q\}$, become 
\ba\label{5-strandHZ}
&&Z(\hat S_5)=\frac{q^{-w} \lambda}{\mathcal{D}_5},\hskip 2mm
Z(\hat S_{41})= \frac{q^{-2w}\lambda^2(q^2+q^{-2})(q+q^{-1})}{\mathcal{D}_5},\nonumber\\
&&Z(\hat S_{32})=\frac{q^{-2w}\lambda^2(q^{-1}+q)+q^{-3w}\lambda^3 (q^{-2}+1+q^2)}{\mathcal{D}_5},\nonumber\\
&&Z(\hat S_{311})=\frac{q^{-3w}\lambda^3 (q^2+q^{-2})(q^2+1+q^{-2})}{\mathcal{D}_5},\nonumber\\
&&Z(\hat S_{221})= \frac{q^{-3w}\lambda^3(q^{-2}+1+q^2)+q^{-4w}\lambda^4(q^{-1}+q)} {\mathcal{D}_5},
\nonumber\\
&&Z(\hat S_{2111})=\frac{q^{-4w} \lambda^4 (q^2+q^{-2})(q+q^{-1})}{\mathcal{D}_5}, \hskip 2mm
Z(\hat S_{11111})= \frac{q^{-5w}\lambda^5}{\mathcal{D}_5},
\ea
with $\mathcal{D}_5=(1-q^{-w-5} \lambda)(1-q^{-w-3} \lambda)(1-q^{-w-1}\lambda)(1-q^{-w+1}\lambda)(1-q^{-w+3}\lambda)(1-q^{-w+5}\lambda)$. 
The Racah matrices for $[41]$ become
\ba
&&
U^{[41]}=\begin{pmatrix} 1 & & &\\
&1& & \\
&&c_2&s_2\\
&&-s_2&c_2
\end{pmatrix},\;\;V^{[41]}=\begin{pmatrix} 1 & & &\\
&c_3&s_3 & \\
&-s_3&c_3&\\
&&&1
\end{pmatrix},\nonumber\\
&&
\hspace{22mm}W^{[41]}=\begin{pmatrix} c_4 &s_4 & &\\
-s_4&c_4& & \\
&&1&\\
&&&1
\end{pmatrix},
\ea
with $c_n=\frac{1}{[n]_q}, s_n=\frac{\sqrt{[n]^2_q-1}}{[n]_q}=\frac{\sqrt{[n-1]_q[n+1]_q}}{[n]_q}, [n]_q=(q^n-q^{-n})/(q-q^{-1}).$ Using these, one can determine the $R-$matrices  by
 \ba\label{RUVW}
&&R_2^{Q}=U^{Q}R_1^{Q}(U^{Q})^T,\;\;R_3^{Q}=U^{Q}V^{Q}R_2^{Q}(U^{Q}V^{Q})^T\nonumber\\
&&R_4^{Q}=U^{Q}V^{Q}W^{Q}R_3^{Q}(U^{Q}V^{Q}W^{Q})^T
\ea
or, in terms  of $R_1$, by
\ba
R_2^{Q}&=& U^{Q} R_1 (U^{Q})^{T}, \;\;R_3^{Q}= (U^{Q}V^{Q}U^{Q})R_1^{Q} (U^{Q}V^{Q}U^{Q})^{T}\\
R_4^{Q}&=&(U^{Q}V^{Q}W^{Q}U^{Q}V^{Q}U^{Q})R_1^{Q}(U^{Q}V^{Q}W^{Q}U^{Q}V^{Q}U^{Q})^{T}.\nonumber
\ea
where $^T$ means the transpose.
Explicitly, for $Q=[41]$, these are
\ba\label{R5-strand}
R_1^{[41]}&=&  \begin{pmatrix} q & & &\\
&q& & \\
&&q&\\
&&&-\frac{1}{q}
\end{pmatrix},\;\;
R_2^{[41]}=\begin{pmatrix} q & & &\\
&q& & \\
&&-\frac{1}{q(1+q^2)}&-\frac{\sqrt{1+q^2+q^4}}{1+q^2}\\
&&-\frac{\sqrt{1+q^2+q^4}}{1+q^2}&
\frac{q^3}{1+q^2}
\end{pmatrix},\nonumber\ea
\ba
 R_3^{[41]}&=& 
\begin{pmatrix} q & & &\\
&-\frac{1}{q(1-q+q^2)(1+q+q^2)}&-\frac{(1+q^2)\sqrt{1+q^4}}{(1-q+q^2)(1+q+q^2)} & \\
&-\frac{(1+q^2)\sqrt{1+q^4}}{(1-q+q^2)(1+q+q^2)}&\frac{q^5}{(1-q+q^2)(1+q+q^2)}&\\
&&
& q
\end{pmatrix},\nonumber\\
R_4^{[41]}&=& \begin{pmatrix}  
-\frac{1}{q(1+q^2)(1+q^4)}&-\frac{q^3\sqrt{-1+\frac{(q^4-q^{-4})^2}{(q-q^{-1})^2}}}{(1+q^2)(1+q^4)}&&\\
-\frac{q^3\sqrt{-1+\frac{(q^4-q^{-4})^2}{(q-q^{-1})^2}}}{(1+q^2)(1+q^4)}&\frac{q^7}{(1+q^2)(1+q^4)}&&\\
&&q&\\
&&&q
\end{pmatrix},
\ea
which satisfy 
\be
{\rm det}\left(R_l^{[41]}\right)= -q^2,\;\;\tr \left(R_l^{[41]}\right)^n = 3 q^n+ (-1)^n \frac{1}{q^n};\;\;l=1,2,3,4.
\ee
For the Young diagram $Q=[32]$, which has two rows, the Racah matrices are the  $5\times 5$ matrices that can be determined from representations of $SU_q(2)$ \cite{MorozovCH}  
\ba\label{Racah32}
&&
U^{[32]}=\begin{pmatrix} 1 & &\\
&U^{(2)}& \\
&&U^{(2)}
\end{pmatrix},\;\; 
V^{[32]}=\begin{pmatrix}V^{(2)} & &\\
&1& & \\
&&1&\\
&&&-1
\end{pmatrix},\nonumber\\
&&
\hspace{14mm}W^{[311]}=\begin{pmatrix} 1 & &&& \\
&-c_2 &&s_2& \\
&&-c_2 &&s_2\\
&s_2&&c_2 &\\&&s_2&&c_2
\end{pmatrix},
\ea 
where $ U^{(2)}$ and $ V^{(2)}$ are the block matrices \be
 U^{(2)}:=\begin{pmatrix}
    -c_2&-s_2 \\
s_2&-c_2\end{pmatrix},\;\; V^{(2)}:=\begin{pmatrix}
     -c_3 &s_3\\s_3&c_3 \end{pmatrix},
     \ee
while
\be\label{R132}
R_1^{[32]}= \begin{pmatrix} q & & & &\\&q&&&\\
&&-q^{-1}& & \\
&&&q&\\
&&&&-q^{-1}
\end{pmatrix}.
\ee

The Young diagram $Q=[311]$ has 3 rows, 
and hence the representation theory of the quantum  group $SU_q(3)$ is now required  to determine the Racah matrices \cite{MorozovCHII}. 
Explicitly, they become the  $6\times 6$ matrices 
\ba\label{Racah311}
&&
U^{[311]}=\begin{pmatrix} 1 & &&\\
&U^{(2)}&&  \\
&&U^{(2)}&\\
&&&1
\end{pmatrix},\;\;V^{[311]}=\begin{pmatrix}V^{(2)} & &\\
&1& & \\
&&-1&\\
&&&V^{(2)}
\end{pmatrix},\nonumber\\&&
\hspace{15mm}W^{[311]}=\begin{pmatrix} -1 & &&& &\\
&-c_4 &&s_4& & \\
&&-c_4 &&s_4&\\
&s_4&&c_4 &&\\&&s_4&&c_4&\\&&&&&1
\end{pmatrix},
\ea 
     and \be\label{R1311}
R_1^{[311]}= \begin{pmatrix} q & & & &&\\&q&&&&\\
&&-q^{-1}&& & \\
&&&q&\\
&&&&-q^{-1}&\\&&&&&-q^{-1}
\end{pmatrix}.
\ee

\vskip 2mm
\noindent{Example $8_3$.} 
Using the braid representative $\sigma_1^2\sigma_2\sigma_1^{-1}\sigma_3^{-1}\sigma_2\sigma_3^{-1}\sigma_4^{-1}\sigma_3\sigma_4^{-1}$ (c.f \cite{KnotInfo}), with $w=0$, and the above matrices, we compute  $h^{[5]}=h^{[11111]}=1$, 
 $h^{[41]}= h^{[2111]}= q^6-q^4-q^2+1-q^{-2}-q^{-4}+q^{-6}$, $h^{[32]}=(q^{-2}-1+q^2)(q-q^{-1})^2$ and $h^{[311]}=-2 q^6 + 3 q^4-q^2 + 1 - q^{-2}+ 3 q^{-4} - 2 q^{-6}$.

\begin{prop}\label{prop:5condition}
The HZ-factorisability  conditions for knots that can be obtained as closed 5-strand braids  are 
\ba
h^{[41]}&=&-\frac{q^{\gamma_0}+q^{\gamma_1}+q^{\gamma_2}+q^{\gamma_3}}{(q^{2}+q^{-2})(q+q^{-1})};\;\; \gamma_i\in\mathbb{Z}_{\rm odd},\;\; {\rm s.t.} \sum_{i=0}^3\gamma_i=2w\nonumber\\
h^{[32]}&=&0,\;\;h^{[311]}=\frac{\sum_{(i,j)}q^{\gamma_i+\gamma_j-w}}{(q^2+q^{-2})(q^2+1+q^{-2})},
\ea
where $(i,j)$ refers to the 6 permutations of $i,j\in\{0,..,3\}$ with $i\neq j$.
A special case of the latter is
$h^{[311]}=\frac{q^{\delta}+q^{-\delta}}{q^2+q^{-2}}$
with   $\delta=\frac{1}{2}(\gamma_0+\gamma_2-\gamma_1-\gamma_3)$ and $\gamma_2=\gamma_1+2$, $\gamma_3=\gamma_2+2$. 
These yield
\be\label{5fact}
Z=\frac{\la}{\mathcal{D}_5}\prod_{i=0}^{3}(1-\la q^{\al_i});\;\;\al_i:=\gamma_i-2w.
\ee
\end{prop}
\begin{proof}
   The HZ function is computed via $Z=\sum_Qh^QZ(\hat{S}_Q)$,  where $Z(\hat{S}_Q)$ are given in (\ref{5-strandHZ}). Substituting the above Racah coefficients together with $h^{[5]}=q^w$ and $h^{[221]}$, $h^{[2111]}$ and $h^{[11111]}$, which  are obtained  from $h^{[32]},h^{[41]}$ and $h^{[5]}$, respectively, by $q\to -q^{-1}$ and using the fact that $w$ is even for \emph{knots} with 5 strands,  the result in (\ref{5fact}) is easily obtained.
\end{proof}
\vskip2mm
\noindent{Remark.}
By Prop.~\ref{prop:3fact}, ~\ref{prop:4fact} and ~\ref{prop:5condition} we observe that HZ-factorisability for knots occurs when only hook-shaped Young diagrams  have a non-vanishing contribution.

\begin{prop}\label{prop:F5E5}
    For full twists $F_5=(\sigma_4\sigma_3\sigma_2\sigma_1)^5$ we compute  $(R_4^{Q}R_3^{Q}R_2^{Q}
    R_1^{Q})^{5}$ for
    \be
    Q=[41]: q^{10}\mathbb{I}_4,\;\;Q=[32]: q^{4}\mathbb{I}_5,\;\; Q=[311]: \mathbb{I}_6,
    \ee
    where, as before, $\mathbb{I}_n$ is the $n\times n$ identity matrix. 
    For a partial full twist $F_4=(\sigma_3\sigma_2\sigma_1)^4$ the matrix $(R_3^QR_2^{Q}R_1^{Q})^{3}$ can be written in block matrix form for each $Q$ as
       \be
[41]:\begin{pmatrix} q ^{12}& \\
& q ^{4}\mathbb{I}_3 \\
\end{pmatrix},\;[32]:\begin{pmatrix} q ^{4}\mathbb{I}_3& \\
&\mathbb{I}_2 \\
\end{pmatrix},\;[311]:\begin{pmatrix} q ^{4}\mathbb{I}_3& \\
& q ^{-4}\mathbb{I}_3 \\
\end{pmatrix}.
    \ee
    For Jucys-Murphy twists $\tilde{E}_5=\sigma_4\sigma_3\sigma_2\sigma_1^2\sigma_2\sigma_3\sigma_4$, $(R_4^{Q}\cdots R_1^{Q}R_1^{Q}\cdots R_4^{Q})$ gives
    \be
[41]:\begin{pmatrix} q ^{-2}& \\
& q ^{6}\mathbb{I}_3 \\
\end{pmatrix},\;[32]:\begin{pmatrix} \mathbb{I}_3& \\
&q ^{4}\mathbb{I}_2 \\
\end{pmatrix},\;[311]:\begin{pmatrix} q ^{-4}\mathbb{I}_3& \\
& q ^{4}\mathbb{I}_3 \\
\end{pmatrix}
    \ee 
    For the combination $F_4^{ j}\otimes F_5^l\otimes \tilde{E}_5^k$ the $R-$matrix representations become
\ba
&&[41]:\begin{pmatrix} q ^{12j+10l-2k}& \\
& q ^{4j+10l+6k}\mathbb{I}_3 \\
\end{pmatrix},\;[32]:\begin{pmatrix}q^{4j+4l} \mathbb{I}_3& \\
&q ^{4k+4l}\mathbb{I}_2 \\
\end{pmatrix},\nonumber\\
&&\hspace{23mm}[311]:\begin{pmatrix} q ^{4j -4k}\mathbb{I}_3& \\
& q ^{-4j+4k}\mathbb{I}_3 \\
\end{pmatrix}
    \ea
\end{prop}
\begin{proof}
    By direct computation using and (\ref{R5-strand}), (\ref{Racah311}), and (\ref{R1311})-(\ref{R132}).
\end{proof}
By Prop.~\ref{prop:F5E5}, concatenating 
 $F_4^{ j}\otimes F_5^l\otimes \tilde{E}_5^k$ to a base braid $\mathbf{b}$ with  $h^Q(\mathbf{b})=\sum_ix_{ii}^Q$, has the following effects on the Racah coefficients
 \ba\label{F5F4E5effect}
 &&h^{[5]}=q^w\to q^{12j+20l+8k}h^{[5]}\nonumber\\
 &&h^{[41]} 
 \to q^{12j-2k+10l}x_{11}^{[41]}+q^{4j+10l+6k}\sum_{i=2}^4x_{ii}^{[41]}\nonumber\\
 &&h^{[32]}\to q^{4(j+l)}\sum_{i=1}^3x_{ii}^{[32]}+q^{4(k+l)}\sum_{i=4}^5x_{ii}^{[32]}\nonumber\\
  &&h^{[311]}\to q^{4(j-k)}\sum_{i=1}^3x_{ii}^{[311]}+q^{-4(j-k)}\sum_{i=4}^6x_{ii}^{[311]}.
 \ea
 For instance the effect of Jucys-Murphy twist $\tilde{E}_5^k$ alone on $h^{[311]}$ is
\be
h^{[311]}\to (-1)^k\left(h^{[311]}+\sum_{i=1}^k(-1)^i(q^{4i}+q^{-4i})\right).
\ee
\noindent{ \bf $\bullet$ Base braid $\mathbf{b}=\sigma_4\sigma_3\sigma_2\sigma_1$. } For this 5-strand braid whose closure is the unknot we compute   $x^{[41]}_{11}=\frac{-q^{-1}}{[4]_q}$, $x^{[41]}_{22}=\frac{-q^{3}}{[4]_q[3]_q}$, $x^{[41]}_{33}=\frac{-q^{3}}{[2]_q[3]_q}$, $x^{[41]}_{44}=\frac{-q^{3}}{[2]_q}$, and hence for $\mathcal{K}^{(5)}_{j,k,l}:=\mathbf{b}\otimes F_4^j\otimes F_5^l\otimes \tilde{E}_5^k$, with $w=4(1 + 3 j + 5 l + 2 k)$, we find
\be
h^{[41]}=-\frac{q^{12j-2k+10l-1}+q^{4j+6k+10l+1}+q^{4j+6k+10l+3}+q^{4j+6k+10l+5}}{(q^2+q^{-2})(q+q^{-1})}.
\ee
Similarly
  $x^{[32]}_{11}=\frac{-1}{[3]_q}$, $x^{[32]}_{22}=\frac{1}{[3]_q[2]_q^2}$, $x^{[32]}_{33}=\frac{1}{[2]_q^2}$, $x^{[32]}_{44}=\frac{-q^{2}}{[2]_q^2}$, $x^{[32]}_{55}=\frac{q^{2}}{[2]_q^2}$, 
  and
$x^{[311]}_{11}=\frac{q^{-2}}{[3]_q}$, $x^{[311]}_{22}=\frac{q^{-2}}{[2]_q^2[3]_q([3]_q-1)}$, $x^{[311]}_{33}=\frac{q^{-2}}{[2]_q^2([3]_q-1)}$, $x^{[311]}_{44}=\frac{q^{2}}{[2]_q^2([3]_q-1)}$, $x^{[311]}_{55}=\frac{q^{2}}{[2]_q^2[3]_q([3]_q-1)}$, $x^{[311]}_{66}=\frac{q^{2}}{[3]_q}$,
with which we compute for $\mathcal{K}^{(5)}_{j,k,l}$
  \be
 h^{[32]}=0,\;\; h^{[311]}=\frac{q^{4j-4k-2}+q^{4k-4j+2}}{(q^2+q^{-2})}.
  \ee 
These Racah coefficients  satisfy the HZ-factorisability conditions in Prop.~\ref{prop:5condition} and hence the HZ transform is of the form (\ref{5fact}) with $\al_0=-3(3+4j+6k+10l)$ and $\al_i=-10(2j+k+3l)-9+2i$ for $i=1,2,3$.

\vskip2mm
\noindent{Example  $T(5,n)$.}
For  5-strand torus knots $T(5,n)$, for $ n\ne 5 \hskip 1mm (mod \hskip 1mm 5)$
\be\label{5full}
h^{[41]}=\tr\left(R_1^{[41]}R_2^{[41]}R_3^{[41]}R_4^{[41]}\right)^n=- q^{2n} \ee 
and the HZ character expansion, in agreement with \cite{MorozovU}, is 
\be\label{HT(5,n)}
H(T(5,n))=A^{-4n}(q^{4n}S_5^{*} - q^{2n}S_{41}^{*}+S_{311}^{*}-q^{-2n}S_{2111}^{*} +q^{-4n}S_{11111}^{*}).
\ee
When $n=5l+1$ this correspond to $\mathcal{K}^{(5)}_{s,s,l-s}$, while $n=5l-1$ to $\mathcal{K}^{(5)}_{l-s,l-s-1,s}$. For $n=5l\pm2$ a different base is required. In particular, $(\sigma_4\sigma_3\sigma_2\sigma_1)^2\otimes F_4^s\otimes F_5^{l-s}\otimes \tilde{E}_5^s$ corresponds to torus knots $T(5,5l+2)$. For general $j,k,l$ $(\sigma_4\sigma_3\sigma_2\sigma_1)^2\otimes F_4^j\otimes F_5^l\otimes \tilde{E}_5^k$ has Racah coefficients 
\ba
h^{[41]}&=&-\frac{q^{12j-2k+10l+3}+q^{4j+6k+10l+1}+q^{4j+6k+10l+5}+q^{4j+6k+10l+7}}{(q^2+q^{-2})(q+q^{-1})},\nonumber\\
&& h^{[32]}=\frac{q^{1+4(j+l)}-q^{1+4(k+l)}+q^{3+4(j+l)}-q^{3+4(k+l)}}{q+q^{-1}},\nonumber\\
 &&\hspace{17mm}h^{[311]}=\frac{q^{4(j-k)+1}+q^{4(k-j)-1}}{(q+q^{-1})},
  \ea
  which satisfy the HZ factorisability conditions only when $j=k$. Similarly, the braid $(\sigma_4\sigma_3\sigma_2\sigma_1)^{-2}\allowbreak\otimes F_4^s\otimes F_5^{l-s}\otimes \tilde{E}_5^s$ gives rise to the torus knots $T(5,5l-2)$.
When $n=5l$ these correspond to full twists $F_5^l=T(5,5l)=(\sigma_4\sigma_3\sigma_2\sigma_1)^{5l}$, which have the HOMFLY--PT expansion
\ba
H(F_5^l)&=&A^{-20l}\big(q^{20l} S^*_{[5]}+4q^{10l}S^*_{[41]}+5q^{4l}S^*_{[32]} + 6 S^*_{[311]}\nonumber\\
&&+ 5q^{-4l}S^*_{[221]} +4q^{-10l}S^*_{[2111]}+ q^{-20l} S^*_{[11111]}\big).
\ea 
Similarly the closed Jucys--Murphy braid $E_5^k$ has HOMFLY--PT character expansion 
\ba
H(E^k_5)\hspace{-2mm}&=&\hspace{-2mm}A^{-8k}(q^{8k}S^*_{[5]}+ q^{-8k} S^*_{[11111]}+(q^{-2k}+3q^{6k})S^*_{[41]}+(3+2q^{4k})S^*_{[32]}\nonumber\\
&&\hspace{-6mm}+3(q^{-4k}+q^{4k})S^*_{[311]}+ (3+2q^{-4k})S^*_{[221]} +(q^{2k}+3q^{-6k})S^*_{[2111]}).
\ea

\vskip2mm
\noindent{Remark.} It is important to comment on the computational efficiency of the HOMFLY--PT polynomial via the character expansion  as opposed to the skein relation. This is due to the fact that, while the skein relation demands a combinatorial computation involving $2^c$ steps ($c$ is the number of crossings) that should be carried all at once; the character expansion involves the multiplication of $c$  matrices, which can be simplified in a straightforward way by splitting the product into sub-products with less than  $c$ matrices and then multiplying them together. Hence, the HOMFLY--PT polynomial of knots with high number of crossings can be computed very efficiently via characters. However, the computational complexity  increases with the number of strands, as the $R-$matrices involved become larger. 
\paragraph{6 strands.}
The Schur functions  $S_Q^{*}$ are 
\ba
&& S_6^{*}= \frac{\{A q\}\{A q^2\}\{A q^3\}\{A q^4\}\{A q^5\}}{\{q^2\}\{q^3\}\{q^4\}\{q^5\}\{q^6\}},\;\;S_{51}^{*}= \frac{\{A /q\}\{A q\}\{A q^2\}\{A q^3\}\{A q^4\}}{\{q\}\{q^2\}\{q^3\}\{q^4\}\{q^6\}},\nonumber\\
&&\hspace{25mm} S_{411}^{*} = \frac{ \{A/q^2\}\{A/q\}\{A q\}\{A q^2\}\{A q^3\}}{\{q\}\{q^2\}^2\{q^3\}\{q^6\}},\nonumber\\
&& S_{42}^{*} =\frac{\{A/q\}\{A\} \{A q\}\{A q^2\}\{A q^3\}}{\{q\}\{q^2\}^2\{q^4\}\{q^5\}},\;\;{S}_{33}^{*}=\frac{ \{A\}\{A  q\}^2\{A /q\}\{A q^2\}}{\{q^2\}^2\{q^3\}^2\{q^4\}},\nonumber\\
&& \hspace{23mm} S_{321}^{*}= \frac{A^w \{A\}^2\{A q\}\{A q^2\}\{A q^{-1}\}\{A q^{-2}\}}{\{q\}^3\{q^3\}^2\{q^5\}}
\ea
and the HZ transform  of $\hat S_Q=A^{-w}\frac{\{A\}}{\{q\}}S^*_Q$
 is
\ba\label{HZ6-strand}
&&Z(\hat S_{6}) = \frac{\lambda q^{-w}}{\mathcal{D}_6},
\hskip 3mm
Z(\hat S_{51}) = \frac{\lambda^2 q^{-2w}(q^4+q^2+1+q^{-2}+q^{-4})}{\mathcal{D}_6},\nonumber\ea
\ba
&&Z(\hat S_{42})=\frac{(q^2+1+q^{-2})(\lambda^2 q^{-2w}+ \lambda^3 q^{-3w}(q^2+q^{-2}))}{\mathcal{D}_6},\nonumber
\\
&&Z(\hat S_{411})=\frac{\lambda^3 q^{-3w}(q^6 + q^4 + 2 q^2 +2+2 q^{-2}+q^{-4}+q^{-6})}{\mathcal{D}_6},\nonumber\\
&&Z(\hat S_{33})= \frac{\lambda^2 q^{-2w}+ q^{-3w}\lambda^3 (q^2+1+q^{-2}) + q^{-4w} \lambda^4}{\mathcal{D}_6},\nonumber\\
&&Z(\hat S_{321})=\frac{q^{-3w} \lambda^3 (q^4+2 q^2+2+2 q^{-2}+q^{-4})(1+ q^{-w} \lambda)}{\mathcal{D}_6},
\ea
where $\mathcal{D}_6= (1- q^{-w-6}\lambda)(1-q^{-w-4}\lambda)(1-q^{-w-2} \lambda)(1-q^{-w}\lambda)(1-q^{-w+2}\lambda)(1-q^{-w+4}\lambda)(1-q^{-w+6} \lambda)$ and $Z(\hat S_{111111}), Z(\hat S_{21111}),Z(\hat S_{3111}),Z(\hat S_{2211}),Z(\hat S_{222})$ which have increasing powers of $\la$, are obtained from (\ref{HZ6-strand}) by the row $\rightleftharpoons$ column symmetry. 
\begin{prop}
    The HZ function for $m=6$ is factorised when
    \ba
  &&  h^{[51]}=-\frac{\sum_{i=0}^4q^{\gamma_i}}{q^4+q^2+1+q^{-2}+q^{-4}};\;\;\sum_{i=0}^{4}\gamma_i=3w\nonumber\\
    &&h^{[411]}=\frac{\sum_{(i,j)}q^{\gamma_i+\gamma_j-w}}{q^6 + q^4 + 2 q^2 +2+2 q^{-2}+q^{-4}+q^{-6}},\nonumber\\
&&h^{[42]}=h^{[33]}=h^{[321]}=0,
    \ea
    where $(i,j)$ refers to the 10 possible permutations for $i,j\in\{0,..,4\}$ with $i\neq j$.
\end{prop}
\begin{proof}
    The proof is similar to the lower strand cases (see e.g. Prop.~\ref{prop:4fact}).
\end{proof}
\paragraph{Higher strands.}
The Schur functions  $S_Q$ for Young diagrams with $m$ boxes are defined in (\ref{S_Q^{hat}}). For the first two Young diagrams the corresponding HZ transform is
\ba\label{mstrandZ}
&&Z(\hat S_{m}) = \frac{\lambda q^{-w}}{\mathcal{D}_m},\;\;Z(\hat S_{[(m-1)1]})= \frac{\lambda^2q^{-2w}}{\mathcal{D}_m} \sum_{i=0}^{m-2} q^{m-2-2i};\nonumber\\
&&\hspace{15mm}\mathcal{D}_m= \prod_{i=0}^m(1- q^{-w-m+2i}\lambda).\ea

For full twists $F_m^l=T(m,ml)$ and Jucys-Murphy twists $E_m^k$, $h^{[(m-1)1]}$ can be obtained as generalisation from the cases with $m\leq 5$ as $h^{[(m-1)1]}(F_m^l)=(m-1)q^{(m-3)ml}$ and  $h^{[(m-1)1]}(E_m^k)=(m-2)q^{2(m-2)k}+q^{-2k}$.  The effect of concatenation with full twists $F_m^l$ is
\be
h^{[m]}\to q^{m(m-1)l}h^{[m]},\;\;h^{[(m-1)1]}\to q^{m(m-3)l}h^{[(m-1)1]},
\ee
while they leave the symmetric Young diagrams (such as $[21]$, $[22]$, $[311]$, etc.) unchanged.

 \begin{cor}\label{cor:h}
    For HZ-factorisable knots with an $m-$strand braid representative the HOMFLY--PT polynomial is fully determined by just the writhe $w$ of the braid diagram and the Racah coefficient
    \be
    h^{[(m-1)1]}=-\frac{\sum_{i=0}^{m-2}q^{\g_i}}{[m-1]_q};\;\gamma_i\in\mathbb{Z}\;\;{\rm s.t.}\;\sum_{i=0}^{m-2}\g_i=(m-3)w.
    \ee
\end{cor}
\begin{proof}
    The HOMFLY--PT polynomial can be obtained by the inverse HZ transform \cite{Petrou1}, and in HZ factorisable cases it is fully determined by the numerator and denominator exponents $\al_i=\g_i-2w$ for $i\in\{0,..,m-2\}$ and $\beta_i=-w-m+2i$ for $i\in\{0,..,m\}$. \end{proof}

\vskip1mm
\noindent{Remark.}   It is noteworthy that the writhe of an $m$-strand braid diagram, can be thought of as an invariant of its closure when $m$ is fixed. This is because the writhe can only be changed by the first Reidemeister move, which can not be introduced to a braid diagram without changing the number of its strands. 
\vskip2mm
We define the general $m-$strand family of knots
\be\boxed{\mathcal{K}^{(m)}_{j,k,l}=\sigma_{m-1}\sigma_{m-2}...\sigma_1\otimes F_{m-1}^j\otimes \tilde{E}_m^k\otimes F_m^l,}\ee
which is HZ-factorisable and serves as a hyperbolic extension of torus knots. It has 
Racah coefficient
\ba\label{h[m-11]}
h^{[(m-1)1]}=-\frac{1}{[m-1]_q}\hspace{-3mm}&&\hspace{-2mm}\bigg(q^{j(m-1)(m-2)+lm(m-3)-2k-1}\nonumber\\
&&\hspace{-5mm}+\sum_{i=1}^{m-2}q^{j(m-1)(m-4)+lm(m-3)+2k(m-2)-1+2i}\bigg),
\ea
which together with the writhe $w=(m-1)(1+lm+2k+(m-2)j)$ are used to determine the numerator and denominator exponents of the HZ transform as
\ba\label{FactGeneralExponents}
\beta_i&=&-(m-1)(1+lm+2k+(m-2)j)-m+2i;\;\;i\in\{0,..,m\}\nonumber\\
\al_0&=&-j(m-1)(m-2)-lm(m+1)-2k(2m-1)-2m+1\nonumber\\
\al_i&=&-jm(m-1)-lm(m+1)-2k(m-2)-2(m-i)+1;
\ea
$i\in\{1,..,m-2\}$.
The HOMFLY--PT polynomial for this family is expressed in terms of these  exponents as \cite{Petrou1}
\be\label{FactorisedHOMFLY}
H({\mathcal{K}};a,q)= \frac{q-q^{-1}}{a-a^{-1}}\sum_{j=0}^m \frac{\prod_{i=0}^{m-2}(1-q^{\alpha_i-\beta_j})}{\prod_{i=0,i\neq j}^{m}(1-q^{2(i-j)})}q^{-\beta_j}a^{\beta_j}.
\ee

\vskip 2mm
\noindent{Example  $T(m,m+1)\otimes E_m^k=\mathcal{K}^{(m)}_{0,k,1}$.} Using (\ref{h[m-11]}) and $w=(m^2-1)+2m(m-1)k$   the numerator exponents 
in the factorised HZ become  $\al_0=-2(2m-1)k- m(m+3)+1$ and $\al_i=-2mk-m(m+1)-1 - 2 i$ for $i\in\{1,..,m-2\}$, which verifies (up to $q\to q^{-1}$) the result  computed in 
\cite{Petrou1}. 
The coefficient $h^{[(m-1)1]}$ can be expressed in polynomial form as
\be h^{[(m-1)1]}=-q^{\delta}+\sum_{i=1}^k(q^{\delta+2i+2(i-1)(m-2)}-q^{\delta+2i(m-1)}),\ee
where $\delta=m(m-2)-2k-3$. 
For instance,  for $T(6,7)\otimes E_6$ ($w=45$), it reads
 $h^{[51]}=-q^{19}+q^{21}-q^{29}$.

\vskip 2mm
\noindent{Example  $T(m,n)$.} For torus knots and links $T(m,n)$ with $n\neq mk$, with braid $(\sigma_{m-1}...\sigma_1)^n$, $w=(m-1)n$, including $\mathcal{K}^{(m)}_{s,s,l-s}=T(m,ml+1)$ and $\mathcal{K}^{(m)}_{l-s,l-s-1,s}=T(m,ml-1)$, we find $h^{[(m-1)1]}=-q^{(m-3)n}$. This yields $\al_i=-n(m+1)-m+2+2i$ for $i\in\{0,..,m-2\}$ and $\beta_i=-n(m-1)-m+2i$ for $i\in\{0,..,m\}$, in agreement with \cite{Petrou1}.
Only single hook type Young diagrams contributions occur for torus \emph{knots} but not necessarily for links, for which the HZ function is not factorisable.

For the HZ-factorisable family $\mathcal{K}^{(m)}_{j,k,l}$ the  Jones and Alexander  polynomials, which are special cases of the HOMFLY--PT, can also be expressed explicitly in terms of the $\al_i$ and $\beta_i$ given above, via inverse HZ transform \cite{Petrou1}. For instance, their Jones polynomial, obtained at $N=2$, becomes
\ba
J(\mathcal{K}^{(m)}_{j,k,l})&=&-q^{-j(m-1)(m-2)-lm(m+1)-2k(2m-1)-2m+1}\nonumber\\
&&-\sum_{i=1}^{m-2}q^{-jm(m-1)-lm(m+1)-2k(m-2)-2(m-i)+1}\nonumber\\
&&+\sum_{i=0}^{m}q^{-(m-1)(1+lm+2k+(m-2)j)-m+2i}.\ea

\subsection*{Jones and Alexander polynomials via characters}
It is of interest to also provide general formulas for the Jones ($N=2$) and Alexander ($N=0$) polynomials in terms of characters.
Since the Racah coefficients $h^Q$,  
are polynomial only in $q$  (i.e., there is no $A=q^N$ dependence), 
they may be related to  the Alexander and Jones polynomials. This was already pointed out  for the 3-strand case in (\ref{h21a}). 
\begin{prop}
\textbf{The character expansion for the Jones polynomial} of a knot or link with an $m$-strand braid representative is 
\be\label{Jexpansion}
J(\mathcal{K};q^2)=q^{-2w}\sum_{i=0}^{\lceil (m-1)/2\rceil}h^{[(m-i)i]}\frac{\{q^{m+1-2i}\}}{\{q^2\}}.
\ee
\end{prop}
\begin{proof}
    The Jones character expansion is obtain by substituting $A=q^2$ in (\ref{Schur}), i.e. $J(q^2)=q^{-2w}\sum_Q h^Q S^*_Q\vert_{A=q^2}$. The Schur functions  (\ref{S_Q^{*}}) become $S_Q^*\vert_{A=q^2}=\frac{\{q\}}{\{q^2\}}\allowbreak\prod_{(i,j)\in Q}\frac{\{q^{i-j+2}\}}{\{q^{h_{i,j}}\}}$ and hence are vanishing for a Young diagram $Q$ that includes a box in the $i^{\rm th}$ column and $j^{\rm th}$ row  that satisfy $i-j+2=0$ (e.g. the box $(1,3)$). Such a box is not included precisely in the  $ \lceil (m+1)/2\rceil$ Young diagrams of the form $Q=[(m-i)i]$, which only consist of two rows.   The Schur function for $Q=[m]=:[m0]$,  evaluates as $S^*_{m}\vert_{A=q^2}=\frac{\{q\}}{\{q^{2}\}} \frac{\{q^{m+1}\}\{q^m\}\{q^{m-1}\}\cdots\{q^3\}\{q^2\}}{\{q^{m}\}\{q^{m-1}\}\{q^{m-2}\}\cdots\{q^2\}\{q\}}=\frac{\{q^{m+1}\}}{\{q^{2}\}}$, 
    while more generally
$S_Q^{*}\vert_{A= q^2}=\frac{\{q^l\}}{\{q^2\}}$, with $l=m+1,m-1,m-3,...,\epsilon$, where $\epsilon$ is $2$ or $1$ when $m$ is odd or even, respectively. All the remaining $S_Q^*$ are vanishing, hence yielding (\ref{Jexpansion}).
\end{proof}
For example, in the  case $m=5$, the  Schur functions in (\ref{S_5}), evaluated at $A=q^2$, become 
$S_{311}^*=S_{221}^*=S_{2111}^*=S_{11111}^*=0$ and $S_5^*= \{q^6\}/\{q^2\}= q^4+ q^{-4}+1$, $S_{41}^*=\{q^4\}/\{q^2\}= q^2+q^{-2}$,  $S_{32}^*=1$. Hence, the Jones polynomial is expressed in this case as
\ba\label{h32Jones}
J(q^2) 
=q^{-2w}(h^{[5]}(q^4+q^{-4}+1) + h^{[41]} (q^2+q^{-2}) + h^{[32]}).
\ea
This implies that, given $h^{[5]}=q^w$ and $h^{[41]}$, then $h^{[32]}$ can be determined via the Jones polynomial. 
For instance, for $8_3$ with $w=0$, using the  values $h^{[5]}=1$ and $h^{[41]}=q^6-q^4-q^2+1-q^{-2}-q^{-4}+q^{-6}$  computed above, along with  the Jones polynomial \cite{Bar-Natan} $J(8_3,q^2) = q^8-q^6 + 2 q^4- 3 q^2 + 3 - 3 q^{-2}+ 2 q^{-4}-q^{-6}+q^{-8}$, (\ref{h32Jones}) gives $h^{[32]}=(q^{-2}-1+q^2)(q-q^{-1})^2$, as expected (c.f. Example below (\ref{R1311})).


\vskip 2mm
\begin{prop}\label{propAlex}
\textbf{ The character expansion for the Alexander polynomial} of a knot 
with an $m$-strand braid representative is
\be\label{Theorem}
\Delta = \frac{1}{\sum_{l=0}^{m-1} q^{m-1- 2l}}
\sum_{{\rm single \hskip 1mm hook}\; Q} (- 1)^{r(Q)+1} h^{Q},
\ee
where $r(Q)$ is the number of rows in the hook-shaped Young diagram  $Q$.

\end{prop}
\begin{proof}
 The Alexander character expansion is obtain by substituting $A=1$ in (\ref{Schur}), i.e. $\Delta(q^2)=\sum_Q h^Q S^*_Q\vert_{A=1}$. 
 The Schur functions $S_Q^*=\frac{\{q\}}{\{A\}}\prod_{(i,j)\in Q}\frac{\{Aq^{i-j}\}}{\{q^{h_{i,j}}\}}$ in (\ref{S_Q^{*}}) are non-vanishing in the limit $A\to 1$ only for single hook diagrams $Q$, since they satisfy $i\neq j$ for any $(i,j)\in Q$, except for $(1,1)$ which gives a factor $\{A\}$ that cancels with the one in the normalisation. 
 When $Q$ is of hook shape, the Schur functions become $S_Q^{*}\big|_{A=1}=\pm S_m^{*}\big|_{A=1}= \pm \frac{\{q\}}{\{q^m\}}$, where the sign is positive for a Young diagram with  only 1 row, and alternates with each additional row. This can be confirmed directly, in the 3-strand case by (\ref{S}),  in the 4-strand case  by (\ref{S_4}), 
 while for 5-strands (\ref{S_5}) yields $S_5^*\big|_{A=1}= \{q\}/\{q^5\}= 1/(q^4+q^2+1+q^{-2}+q^{-4}) = -S_{41}^*\big|_{A=1}=S_{311}^*\big|_{A=1}= -S_{2111}^*\big|_{A=1}= S_{11111}^*\big|_{A=1}$. 
 For non-hook diagrams there exist boxes indexed by $(i,i)$ other than $(1,1)$, resulting in an extra factor $\{A\}$, which vanishes at $A=1$. Hence such diagrams do not appear in the Alexander character expansion (\ref{Theorem}). 
\end{proof}

Similar to the above considerations for the Jones polynomial, it may be concluded that the Alexander polynomial $\Delta(q^2)$, via (\ref{Theorem}), may be used  as an alternative way to determine one of the Racah coefficients $h^Q$, when $Q$ is of hook shape. In particular, by Proposition~\ref{propAlex}, the Alexander polynomial in the $3$-strand cases expressed as
\be\label{78}
\Delta = \frac{q-q^{-1}}{q^3-q^{-3}}( h^{[3]}  - h^{[21]}+ h^{[111]})= \frac{q^w-h^{[21]}+(-q)^{-w}}{ q^2+1+ q^{-2}}.
\ee
For the figure-8 knot, for instance,  $w=0$ and hence 
$h^{[21]}$ is obtained via the Alexander polynomial \cite{Bar-Natan} $\Delta(4_1) =3-q^2-q^{-2}$, 
as $h^{[21]}= q^4- 2 q^2+1 - 2 q^{-2}+ q^{-4}$. 
\vskip1mm

\noindent{}{Example $T(m,n)$.} The expansion in (\ref{Theorem}) is consistent with the factorised formula for Alexander polynomial for torus
knots \cite{Murasugi} $\Delta_{m,n}(t)=t^{-(m-1)(n-1)/2} (1-t)(1-t^{mn})/(1-t^m)(1-t^n)$. For instance, in the case of $T(5,3)$, $h^{[5]}= q^{12}, h^{[41]}= - q^{6}, h^{[311]}= 1, h^{[2111]}= - q^{-6}, h^{[11111]}=q^{-12}$ as obtained by 
(\ref{HT(5,n)}). Hence, (\ref{Theorem}) yields
$\Delta(T(5,3)) = (q^{-12}+ q^{-6}+1 + q^6 + q^{12})/(q^4+q^2+1+q^{-2}+q^{-4})= q^8- q^6+ q^2+q^{-2}- q^{-6} + q^{-8} -1$, which agrees with $\Delta_{m,n}(t)$ at  $m=5,n=3$ and $q^2=t$.

\section{Factorised form decomposition\label{sec:factform}}

In the previous section, solidifying the results of \cite{Petrou1}, we have extensively described knots and links that admit a factorised HZ function in both the numerator and denominator. However, such knots and links are very special, as for the vast majority of knots, the numerator of the HZ transform is not factorisable. On the other hand,  the HZ denominator  universally admits the factorised form
 $\prod_{i=0}^m(1-\lambda q^{\beta_i})$, as can also be seen by the denominator $\mathcal{D}_m$ of $Z(\hat S_Q)$, defined  in (\ref{mstrandZ}). 
 This suggests that  we may group knots into different ``HZ types",  which share the same HZ denominator exponents ($\beta_0,...,\beta_m$). 
 In this section we show that for more general knots and links the HZ function may still admit a relatively simple form (at least when the braid index is small), according to the following conjecture.

 \begin{conj}\label{conj:FactForm}
    The Harer--Zagier transform $Z(\mathcal{K};\lambda,q)$ of the HOMFLY--PT polynomial for any knot $\mathcal{K}$ of HZ type $(\beta_0,...\beta_m)$, can be expressed as the sum of factorised terms 
       \be\label{HZstructure}
    [\alpha_0,...,\alpha_{m-2}]:=\frac{\lambda\prod_{i_l=0}^{m-2}(1-\lambda q^{\alpha_{i}})}{\prod_{i=0}^{m}(1-\lambda q^{\beta_i})}\;\;{\rm s.t.}\;\;\sum_{i=0}^{m-2}\alpha_{i}=\sum_{i=0}^{m}\beta_i
    \ee 
    as 
 \be\label{GeneralHZ}
 \boxed{Z(\mathcal{K};\lambda,q) = \sum_i c_i [\alpha_0,\alpha_1,...,\alpha_{m-2}]_i,}
 \ee
 in which the coefficients satisfy $\sum_i c_i=1$. 
  The HZ transform is fully factorised if the only non-vanishing coefficient is $c_1=1$.
  \end{conj}

\vskip2mm
\noindent{Remark.} For  fixed $m$, for which the writhe of the diagram becomes an invariant of the knot, 
we can write
  \be\label{sumExponents}
  \sum_{i=0}^{m-2}\alpha_{i}=\sum_{i=0}^{m}\beta_i=-(m+1)w.
  \ee
\begin{proof} [Proof of Conj.~\ref{conj:FactForm} at $m=2,3$] For the case $m=2$ the proof  is trivial, since the HZ function (\ref{HZ2strand}) is always factorised. 
 Using the HZ character expansion in the $m=3$ case, 
 the factorised form decomposition is determined by the Racah coefficient $h^{[21]}$. Since the latter is always an alternating symmetric polynomial (since the Young diagram is symmetric), it can be expressed as $h^{[21]}(q)=\sum_{i=0}^\zeta (-1)^i\eta_{-\zeta+2i}q^{-\zeta+2i},$ $\zeta\in\mathbb{Z}$, where the coefficients satisfy $\eta_{-\zeta+2i}=\eta_{\zeta-2i}$ and  $\eta_{\zeta}=\pm1$. 
 Hence the $\mathcal{O}(\la^2)$ term in the numerator of HZ, which is  $q^{-2w}(q+q^{-1})h^{[21]}$ yields the  factorised form decomposition
 \ba\label{FFDh21}
Z(\la,q)\hspace{-2mm}&=&\hspace{-2mm} -\eta_\zeta[-2w-\zeta-1,-2w+\zeta+1]\\\nonumber
&&\hspace{-2mm}- \sum_{i=0}^{\zeta/2-1}
 (-1)^i( \eta_{\zeta-2i}- \eta_{\zeta-2-2i})[-2w-\zeta+1+2i,-2w+\zeta-1-2i],
 \ea
 where the denominator of the bracket $[\cdot,\cdot]$ is $\mathcal{D}_3$, as given in (\ref{D3w}).
 Indeed the paired exponents in each factorised form sum to $-4w=\sum_{i=0}^3\beta_i$ and the coefficients are $c_1=-\eta_\zeta$ and $c_i=(-1)^{i+1}(\eta_{\zeta-2i+4}- \eta_{\zeta-2i+2})$ for $2\leq i\leq \frac{\zeta}{2}+1$, which sum to
 $\sum_{i}c_i=2\sum_{i=0}^{\zeta/2-1}(-1)^{i+1}\eta_{\zeta-2i}+(-1)^{\zeta/2+1}\eta_0=-h^{[21]}|_{q=1}$. At $q=1$ the $R-$matrices become
 \be
 R_1=\begin{pmatrix}
    1&\\&-1 
 \end{pmatrix}=R_1^{-1}, \;\; R_2=\begin{pmatrix}
    -\frac{1}{2}&\frac{\sqrt{3}}{2}\\\frac{\sqrt{3}}{2}&\frac{1}{2} 
 \end{pmatrix}=R_2^{-1},\ee
 which both square to $\mathbb{I}_2$ and satisfy $\tr(R_1R_2)=-1$. Hence for any 3-strand knot, for which the writhe is even and its braid diagram must contain $\sigma_1$ and $\sigma_2$ at least once, $h^{[21]}(q=1)=-1$, implying $\sum_{i}c_i=1$. In the HZ-factorisable cases $\eta_\zeta=-1$ and  $|\eta_i|=1\;\;\forall i$ (c.f. Prop.~\ref{prop:3fact}).
  \end{proof} 
 In the $m=4$ case, the proof for the decomposition is not so straightforward, since the order $\mathcal{O}(\la^2)$  term in the HZ numerator depends on two Racah coefficients. Explicitly, it is given by \be\label{Ol^2}
 q^{-2w}((q^2+1+q^{-2})h^{[31]}+h^{[22]}),\ee which is neither a symmetric nor an alternating polynomial. Each  factorised term contains a triple of integers $\alpha_0,\alpha_1,\alpha_2$
which sum to  $-5w=\sum_{i=0}^{4}\beta_i$, which should match the   exponents of the polynomial at  order $\mathcal{O}(\la^2)$.

In fact, for $m=5$ and $m=6$, the $\mathcal{O}(\la^2)$ term alone is not enough to determine the factorised form decomposition, since it does not contain all the independent Racah coefficients.  Therefore, the $\mathcal{O}(\la^3)$ term should  also  be considered. For instance, for $m=5$, by (\ref{5-strandHZ}) the term of order $\mathcal{O}(\la^2)$ is expressed in terms of Racah coefficients as
\be\label{Ol2-5strand}
q^{-2w}(q+q^{-1})((q^2+q^{-2})h^{[41]}+h^{[32]}),
\ee
while the $\mathcal{O}(\la^3)$ term is
\be\label{Ol3-5strand}
q^{-3w}(q^2+1+q^{-2})((q^2+q^{-2})h^{[311]}+h^{[32]}).
\ee 
Indeed the latter contains  the Racah coefficient $h^{[311]}$, which provides  an extra degree of freedom that does not appear at order $\mathcal{O}(\la^2)$.

 The factorised form decomposition for all knots with up to $8$ crossings and selected 9-crossing knots is included in 
 the Appendix. Below we consider the factorised form decomposition for several examples  for up to $m=5$,  
 grouped together according to their HZ type. 
 
\vskip 2mm
{\bf $\bullet$ $(-3,-1,1,3)$ type}
 \vskip 1mm
\noindent{}Knots of this type have HZ denominator 
$\mathcal{D}_3=(1-q^{-3}\lambda)(1-q^{-1}\lambda)(1-q \lambda)(1- q^{3}\lambda)$. 
The HOMFLY--PT polynomial $H$ for such knots  contains  $A$ only in the powers $\pm2$, since the HZ is computed from the unnormalised version $\bar{H}=(A-A^{-1})/(q-q^{-1})H$ (the overall factor changes the powers of $A$ by $\pm1$, resulting in the aforementioned denominator exponents).
The knot $4_1$, for example, has HOMFLY--PT polynomial 
$
H(4_1) =A^2+ A^{-2} - z^2 -1.
$
In its character expansion the coefficient $h^{[21]}$ is given by (\ref{t1}). After  multiplying it with the factor $(q+q^{-1})$ of $Z(\hat S_{21})$ in (\ref{ZC3_1}), the order $\lambda^2$ term becomes
$(q^4- 2 q^2+ 1- 2 q^{-2}+ q^{-4})(q+q^{-1})= q^5- q^3- q- q^{-1}- q^{-3} +q^{-5}$. 
Hence, its HZ can be decomposed into  the sum of factorised terms 
\be\label{ZZ4_1}
Z(4_1) = -[-5,5]+ [-3,3]+[-1,1]
\ee
  where by (\ref{HZstructure})
$$
[-n,n] = \frac{\lambda(1-q^{-n}\lambda)(1-q^n \lambda)}{(1-q^{3}\lambda)(1-q \lambda)(1-q^{-1}\lambda)(1-q^{-3}\lambda)}.
$$
There are 16 different knots with up to 10-crossings with this HZ type. The majority of them that have braid index $3$ and even number of crossings 
and admit a  braid diagram with writhe   $   w=0$. 
Among them for instance, are  the knots $6_3$ and $8_{9}$, which  have  HZ decompositions  
\ba\label{818}
&&Z(6_3)= [-7,7] - [-5,5] + [-3,3]\nonumber\\
&&Z(8_9) = - [-9,9] + [-7,7]- [-5,5] + 2 [-3,3]  
\ea
Usign braid diagrams, such as the ones presented in \cite{KnotInfo},  we compute the Racah coefficients  $h^{[21]}(6_{3})= - q^{-6} + 2q^{-4} - 3q^{-2}+3 - 3 q^2 + 2 q^4 - q^6$ and
$h^{[21]}(8_{9})=q^{-8} - 2q^{-6} + 3q^{-4} - 5q^{-2}+5 - 5 q^2 + 3 q^4 - 2 q^6 + q^8$. 
These give the paired exponents and  the coefficients in the above decompositions by (\ref{FFDh21}). 

\vskip2mm
\noindent{Remark.} These knots are part of a sequence $4_1\to 6_3 \to 8_9 \to 10_{17}\to 12_{a1273}$, which is obtained  by attaching the braid configuration $\sigma_1^n\sigma_2^{-n}$ successively for each $n\geq0$ to the braid $(\sigma_2^{-1}\sigma_1)^2$, whose closure is $4_1$. 
This relation can be seen in the 
 braid representatives of these knots  in \cite{Bar-Natan}. Their factorised form decomposition can be expressed in terms of $n$ as\footnote{Notably, the members of this sequence are related by
insertion of a $\mathbb{Z}_2$ or a $\mathbb{Z}_2$-lego piece, which connects the $j$ and $j-2$ entry boxes in their Kh tables, in a similar fashion as  explained in Sec.~5 of \cite{Petrou1} for HZ-factorisable cases.} 
\ba
\hspace{-19mm}Z(n)\hspace{-2mm}&=&\hspace{-2mm}-[-5,5]+[-3,3]+[-1,1]\nonumber\\
&&\hspace{-14mm}+\sum_{i=0}^{n-1}(-1)^i\left([-2n-5+2i,2n+5-2i]-[-2n+1+2i,2n-1-2i]\right).
\ea
\vskip2mm

The remaining knots of this type with up to 10 crossings and braid index 3 are $8_{17},8_{18},10_{17},10_{48},10_{79},10_{91},10_{99},10_{104},10_{109},10_{118},10_{123},10_{125}$ for  which the decompositions can be found in Appendix B of \cite{PetrouPHD}. 
It is noteworthy that these knots give an exhaustive list of amphichiral knots with up to 10-crossings and braid index $3$, with the exception of  $10_{48}$, $10_{91}$, $10_{104}$ and $10_{125}$ which are chiral, but the HOMFLY--PT polynomial fails to distinguish them from their mirror image. Note that all the remaining amphichiral knots with up to 10 crossings have braid index $5$ and HZ type $(5,3,1,-1,-3,-5)$, 
which is the next possible option satisfying $\sum_i\beta_i=0$.

The knot  $9_{42}$, which has braid index\footnote{The Morton-Franks-Williams inequality, which states that the ${\rm braid\;\;index}\geq m$ (c.f. \cite{Petrou1}), is not sharp in this case and hence there is a common factor in the HZ numerator which cancels with a term in the denominator.} $4$, may also be thought of as belonging to this type, with HZ  decomposed as
\ba\label{Z(6_3)} 
Z(9_{42}) &=& - [-7,7] + [-3,3] + [-1,1].
\ea
This can be related to that of $Z(4_1)$ in (\ref{ZZ4_1}) by the replacement  $[-1,1] \to [-7,7]$.
Its writhe is and  $w=-1$, while its Racah coefficients are 
$h^{[31]}=-2q^{-3} + 3q^{-1} - 3 q + q^3$ and $h^{[22]}=q^{-5} - q^{-3} + 2q^{-1} - 2 q + q^3 - q^5$, which by (\ref{Ol^2}) yield the symmetric polynomial $q^{-7} - q^{-3} - q^{-1} - q - q^3 - q^5 + q^7$, from which (\ref{Z(6_3)}) is determined. 
From the factorised form decomposition  it is easy to derive  recursion formulas for the HZ of knots that share the same type. 
For instance, from the ones presented above, we observe
\be\label{relation1}
Z(6_3)+Z(9_{42})= Z(4_1)+Z(\bigcirc),\ee
where $Z(\bigcirc)= \lambda/(1-q^{-1}\lambda)(1-q\lambda)$, is the HZ corresponding to the unknot with HOMFLY--PT $\bar H(\bigcirc)=\frac{A-A^{-1}}{z}$, which again  may be though of as included in this HZ type with $Z(\bigcirc)=[-3,3]$.


\vskip 2mm
{\bf $\bullet$ (1,3,5,7) type}
 \vskip 1mm  
\noindent{}The knots $3_1,5_2,8_2,8_{21}, 10_{85},10_{100},10_{126},10_{159}$ (8 knots in total,  up to 10 crossings), which have braid index $3=m$,  belong to  this HZ type. 
The sum of the denominator exponents is $\sum_i\beta_i=16$, which determines the sum of numerator exponents in each factorised form, according to (\ref{HZstructure}). Some examples are
\ba\label{1357}
\hspace{-5mm}Z(5_2) &=& [13,3]\nonumber\\
\hspace{-5mm}Z(8_2) &=&- [17,-1]+[15,1] +[11,5]  \nonumber\\
\hspace{-5mm}Z(10_{100})&=& [19,-3]-2[17,-1] + 3[15,1]
-2 [13,3]+2[11,5]-[9,7]. 
\ea
For $5_2$, which is HZ-factorisable, there is a single term in the decomposition. For the remaining two knots, for which  $w=-4$, the Racah coefficients  are $h^{[21]}(8_2)= q^{-8} - 2q^{-6} + 2q^{-4} - 3q^{-2} +3- 3 q^2 + 2 q^4 - 2 q^6 + q^8$ and $h^{[21]}(10_{100})=- q^{-10} + 3q^{-8} - 6q^{-6} + 8q^{-4} - 10q^{-2} +11 - 10 q^2 + 8 q^4 - 
 6 q^6 + 3 q^8 - q^{10}$, 
 yielding the above decompositions by (\ref{FFDh21}).

\vskip 2mm
{\bf $\bullet$ (-5,-3,-1,1,3) type}
\vskip 1mm

\noindent{}Knots of this type with $m=4$, which have braid index at least $4$, include $6_1$, $7_7$, $8_4$ and $8_{13}$. 
 Using the Racah coefficients in the character expansion in (\ref{CH6_1}), the HZ transform for $6_1$ becomes 
\ba\label{HZ6_1}
Z(6_1) &=& \frac{\lambda}{\mathcal{D}_4}  (1- \lambda (q^{5}- q^{3} - q^{-1} - 2 q^{-3} - q^{-5} + q^{-9}) \nonumber\\
&&+ \lambda^2 (1-q^{4}+ 2 q^{-2}+ q^{-4} + q^{-8} - q^{-10})
 - q^{-5} \lambda^3),
\ea
where 
$\mathcal{D}_4 = (1-q^3 \lambda)(1-q \lambda )(1- q^{-1} \lambda)(1-q^{-3} \lambda)(1-q^{-5} \lambda)$. 
This can be expressed as the sum of factorised forms
\be\label{D6_1}
Z(6_1) = [-5,-3,3] -[-1,5,-9] + [-1,-1,-3],
\ee
with  the  notation  of (\ref{HZstructure}) and $\sum_i\alpha_i=\sum_i\beta_i=-5$. Indeed, the numbers in each factorised form match with the exponents of the polynomial at order $\mathcal{O}(\la^2)$ in the HZ numerator.

The decomposition for 
knots with $m=5$ is easily obtained by the following algorithm. First we express the terms of order $\lambda^2$ in the HZ function  as a sum of factorised forms $\sum_ic_i(\mathcal{O}(\lambda^2))[\alpha_0,\alpha_1,\alpha_2,\alpha_3]_i$ such that $\sum_ic_i(\mathcal{O}(\lambda^2))=1$. This automatically yields exactly the terms of order $\lambda^4$.  The next step is  finding  the necessary corrections to account for the terms of order $\lambda^3$. This is because the factorised forms that result in the terms of order $\lambda^2$ (by summing 2 of their exponents), also give terms of order $\lambda^3$ (by summing 3 of their exponents), but which do not match the correct ones appearing in the HZ transform, due to its dependence on $h^{[311]}$ (c.f. (\ref{Ol3-5strand})). Such corrections are obtained by 
adding extra factorised forms in the form of
quadruples, which are such that they do not affect the  terms of order $\lambda^2$ and $\lambda^4$. This is achieved by making sure that the sum of the coefficients of this quadruple vanish and hence such contributions cancel. Explicitly, the quadruples are of the form  $[a,b,c,d]-[a,b,c',d']+[a',b',c',d']-[a',b',c,d]$, which have a cyclic structure.  

\vskip 2mm
{\bf $\bullet$ (-5,-3,-1,1,3,5) type}
 \vskip 1mm
\noindent{}Knots of this HZ type have at least braid index $5$ and have braid diagrams with vanishing writhe $w=0$. These include, for example, the  knots $8_3$ and $8_{12}$.
As an example, we consider $8_3$ for which the Racah coefficients are 
 \ba
 && 
 h^{[41]} = q^6-q^4 - q^2 + 1 - q^{-2}- q^{-4} + q^{-6}\nonumber\\
 &&
 h^{[32]}= (q^2 - 1+ q^{-2}) (q-q^{-1})^2\nonumber\\
 &&h^{[311]} = - 2 q^6 + 3 q^4 - q^2 + 1 - q^{-2} + 3 q^{-4} - 2 q^{-6}.
 \ea
 Using these and $Z(\hat S_Q)$ in (\ref{5-strandHZ}), $Z(8_3)= \sum h^Q Z(\hat S_Q)$  becomes
 \ba
\hspace{-0.5mm} Z(8_3)\hspace{-2mm}& =&\hspace{-2mm} \frac{\lambda}{\mathcal{D}_5} \bigg(1+ (q^{-9}- 2 q^{-3}-q^{-1}- q - 2 q^3 + q^9)\lambda\nonumber\\
 &&\hspace{-15mm}+ (-1- 2 q^{-{10}}- q^{-8} + 7 q^{-6}- 7 q^{-4}+ 6 q^{-2}+ 8 q^2 - 8 q^4 + 4 q^6 + q^8 - 2 q^{10}) \lambda^2 \nonumber\\
 &&\hspace{-15mm} + (q^{-7}- q^{-3}+ q^{-1}+ q - q^3 + q^7)\lambda^3 + \lambda^4\bigg).
\ea
Following the above proposed algorithm we derive the factorised form decomposition 
\ba\label{Z8_3}
Z(8_3)\hspace{-2mm}&=&\hspace{-2mm}-[9,1,-1,-9]+2[3,1,-1,-3]\nonumber\\&&\hspace{-2mm}+[7,-7,1,-1] -[7,-7,3,-3]
+[5,-5,3,-3]-[5,-5,1,-1],\nonumber
\ea
in which the quadruple corrections appear in  the second line.

\vskip 2mm
{\bf $\bullet$ (-1,1,3,5,7,9) type}
 \vskip 1mm
 \noindent{}Some examples of this type, for which the HZ decomposition is obtained by the same algorithm 
 are
 \ba\label{9_{12}}
\hspace{-4mm} Z(9_{12}) \hspace{-2mm}&=& \hspace{-2mm}-[15,11,-3,1] + 2[13,9,3,-1]+[17,7,1,-1]-[13,11,1,-1]\nonumber\\
 &&\hspace{-6mm}+ ([ 17,-1,9,-1]-[17,-1,7,1]+[11,5,7,1]-[11,5,9,-1]),\nonumber\\
\hspace{-4mm} Z(9_{15})\hspace{-2mm}& = &\hspace{-2mm}[17,3,3,1] + 2[13,9,3,-1]-[15,11,-3,1]-[11,9,3,1] \nonumber\\
 &&\hspace{-6mm}+([ 17,-3,9,1] -[17,-3,7,3]
 +[13,1,7,3]-[13,1,9,1])\nonumber\\
 &&\hspace{-6mm}+ 2([19,7,3,-5]-[19,7,1,-3]+[15,11,1,-3]-[15,11,3,-5]).
 \ea

The proposed algorithm may be generalised for   any $m$, as follows.
First find a partial decomposition $\sum_ic_i(\mathcal{O}(\lambda^2))[\alpha_0,\cdots,\alpha_{m-2}]_i$ with the coefficients satisfying $\sum_ic_i(\mathcal{O}(\lambda^2))=1$,  such that to obtain exactly the terms of order $\lambda^2$ (which involves the Racah coefficients $h^{[(m-1)1]}$ and $h^{[(m-2)2]}$). By doing this, we automatically recover precisely the terms of order $\lambda$, $\lambda^m$ and $\lambda^{m-1}$, plus some factors contributing to other orders of $\lambda$. Subsequently, after computing the remainder $-Z(q,\la)\allowbreak+\sum_ic_i(\mathcal{O}(\lambda^2))[\alpha_0,\cdots,\alpha_{m-2}]_i$, which involves only $\mathcal{O}(\lambda^k)$ terms, with $3\leq k\leq m-2$, we need to add corrections for each increasing order of $\lambda$ up to $\lceil\frac{m}{2}\rceil$, which come in quadruples (for $m=5,6$), octuples (for $m=7,8$), etc., which are such that $\sum_ic_i^{\rm corr}=0$.
 Due to the multiple correction terms for higher strands, such factorised form decompositions become lengthy\footnote{Although omited here for simplicity, explicit examples with  $m=6,7$, including ones with octuple corrections, 
 can be found in Appendix B of \cite{PetrouPHD}. 
 }  and hence they no longer provide an efficient way to describe the HZ function.  However, the fact that it is possible to express it in this way, according Conjecture~\ref{conj:FactForm}, relies on an important hidden structure of the HOMFLY--PT polynomial and its HZ function, which is worth exploring further. 

This can be achieved by expressing the coefficients  of $\lambda^N$ at each $N$ in the HZ function $Z(\mathcal{K})=\sum_{N=0}^{\infty}\bar{H}\lambda^N$, which are nothing else than the unnormalised HOMFLY--PT polynomial, in the form 
\be \label{orientable1app}
 \bar H({\mathcal{K}};A,q)= \sum_{i=0}^m\sum_{g\ge 0}^{g_{\rm max}}\hat{N}_{g,i} z^{2g-1} A^{-e-2i},
 \ee
 where $-e\in\mathbb{Z}$ is the lowest power of $A=q^N$ in $\bar{H}$ and $\hat{N}_{g,i}\in\mathbb{Z}$ are constants (they are the $SU(N)$-BPS invariants \cite{Petrou1}). 
 First, note that the constant term at $N=0$ ($A=1$) corresponding to the unnormalised Alexander polynomial,  is always vanishing due to the overall factor $A-A^{-1}$. 
 Explicitly, this is expressed by
 \be \label{N=0contribution}
 \sum_{i=0}^m\sum_{g\ge 0}^{g_{\rm max}}\hat{N}_{g,i} z^{2g-1}=0.
 \ee
 Moreover, at $N=1$ $(A=q)$, $\bar H$ is always equal to $1$. This can be understood from the character expansion, as the Schur functions $S_Q$ in (\ref{S_Q^{hat}})  vanish for all Young diagrams $Q$  that include a box in the second row, at $(i,j)=(1,2)$,  which yields a factor $\{Aq^{i-j}\}|_{A=q}=\{q^0\}=0$ in the numerator. When $Q$ is made of a single row $S_Q=1$ and $h^Q=q^w$, hence, by (\ref{Schur}), $\bar H(A=q)=1$. 
This implies
  \be\label{eq:SU(1)B}
 \sum_{i=0}^m\sum_{g\ge 0}^{g_{\rm max}}\hat{N}_{g,i} z^{2g-1}q^{-e-2i}=1
 \ee
and, as a result, the coefficient of order $\mathcal{O}(\lambda)$ is always equal to
\footnote{\label{overall sign}
This holds up to an overall $\pm1$,  depending on the normalisation conventions chosen for $\bar H$, from which the HZ function is obtained.}
$1$. This is the reason that the sum of the coefficients in (\ref{GeneralHZ}) must be equal to $1$, i.e. $\sum_i c_i=1$. 

The HZ  transform of (\ref{orientable1app}), which is obtained by $A^k\to (1-\la q^k)$, yields 
\ba \label{ZBPS}
 Z&=& \sum_{i=0}^m\sum_{g\ge 0}^{g_{\rm max}}\frac{\hat{N}_{g,i} z^{2g-1}}{1-\la q^{-e-2i}}\nonumber\\
 &=&\sum_{g\geq 0}z^{2g-1}\left(\frac{\hat{N}_{g,0}}{1-\lambda q^{-e}}+\frac{\hat{N}_{g,1}}{1-\lambda q^{-e-2}}+\cdots+\frac{\hat{N}_{g,m}}{1-\lambda q^{-e-2m}}\right).
 \ea
By multiplying out to get the common denominator $\prod_{i=0}^m(1-\lambda q^{-e-2i})$, the highest order in $\lambda$ in the numerator is $(-1)^{m-1}\lambda^m\sum_{i=0}^m\sum_{g\ge 0}\hat{N}_{g,i} z^{2g-1}q^{-m(e+m-1)+2i}$. Using (\ref{eq:SU(1)B}) with 
$q\mapsto q^{-1}$ (note that both $z^{2g-1}$ and $\hat{N}_{g,i}$ change sign under $q\mapsto q^{-1}$, so that $\bar H(A=q)$ remains unchanged), 
 then this term 
 becomes $(-1)^{m-1}\lambda^m q^{\sum_{i}\beta_i}$, where $\sum_{i=0}^{m}\beta_i=-m(e+m+1)-e$. This justifies the requirement that in each factorised term $[\alpha_0,\alpha_1,...,\alpha_{m-2}]$ the exponents are such that $\sum_{i=0}^{m-2}\alpha_{i}=\sum_{i=0}^{m}\beta_i$, so that together with the condition $\sum_i c_i=1$, (\ref{GeneralHZ}) gives rise to the correct highest order $\mathcal{O}(\lambda^m)$ term. Using (\ref{sumExponents}), the parameter $e$ can be expressed in terms of $w$ and $m$ as $e=w-m$, which agrees with 
 $\mathcal{D}_m$ in (\ref{mstrandZ}).

The up-to-an-overall-factor  
 symmetry between the $q$-polynomials  at orders 
$\mathcal{O}(\lambda^2)$ and $\mathcal{O}(\lambda^{m-1})$, $\mathcal{O}(\lambda^3)$ and $\mathcal{O}(\lambda^{m-2})$ etc., is understood due to the symmetry between the Racah coefficients involved. The fact that the coefficients in these $q$-polynomials are the same can alternatively be explained due to the  binomial expansion, since at $\mathcal{O}(\lambda^k)$, 
each exponent in the the $q$-polynomial is obtained by summing together $k$  denominator exponents $\beta_i$, 
which can be done in $\binom{m+1}{k}$ ways and, hence, it coincides with the coefficient $\binom{m+1}{m+1-k}$ appearing in the $\mathcal{O}(\lambda^{m+1-k})$ polynomial.

 A necessary condition for a partial factorised form decomposition at order $\la^2$ 
s.t. $\sum_ic_i(\mathcal{O}(\lambda^2))=1$ to be possible, is that the sum of all the coefficients of the $q$-polynomial at $\mathcal{O}(\lambda^2)$ in $Z$, which can be obtained by setting $q=1$, should be equal to\footnote{Up to a sign, due to footnote~\ref{overall sign}.}  $m-1$ (this can be verified from e.g. (\ref{ZZ4_1}) and (\ref{HZ6_1})). The latter is the number of exponents that appear in $[\alpha_0,\cdots,\alpha_{m-2}]$, hence the requirement ensures that the factors of $q$ with positive and negative coefficients can be exactly  grouped together  in this way. 
The fact that this holds for any knot can be seen from  $Z(\mathcal{K};\lambda,q=1)=\lambda/(1-\lambda)^2=\lambda(1-\lambda)^{m-1}/(1-\lambda)^{m+1}$  $\forall\mathcal{K}$, which was  proven in Prop.~2.1 of \cite{Petrou1}. Using the binomial expansion $(1-\la)^{m-1}=\sum_{k=0}^{m-1}\binom{m-1}{k}\la^k$, 
 the coefficient of $\mathcal{O}(\lambda^2)$ is $\binom{m-1}{1}=m-1$, as required.

Although the above considerations were mainly focused on knots, which consist of a single component, similar results can be derived for links with multiple components, as shown by explicit examples in the next section. 

\section{Dynkin diagrams, Coxeter links and their HZ functions\label{sec:forestquiv}}
An intriguing connection between the characteristic polynomials of  ADE singularities and the zero distributions of the HZ function 
was presented in \cite{Petrou}. There is yet another interesting way to relate ADE singularities with knots and links. This is achieved by assigning a certain link, called the Coxeter link, to a Dynkin diagram of star shape (i.e. of ADE type)\footnote{The homology of a link, in  relation to  a Dynkin diagram, is well known \cite{ACampo}. By considering the fundamental cycles on a Riemann surface as the vanishing of a polynomial,  one obtains ADE singularities, with each cycle corresponding to a node in the Dynkin diagram.} \cite{Eriko}. Such a link  corresponding to a Dynkin diagram with three `legs' with $p,q,r$ nodes on each (overlapping at a single node), is shown in Fig.~\ref{fig:coxeterlink}. In the present article we view this correspondence in the light of the HZ transform. 
\begin{figure}[h!]
    \centering
    \includegraphics[width=0.5\linewidth]{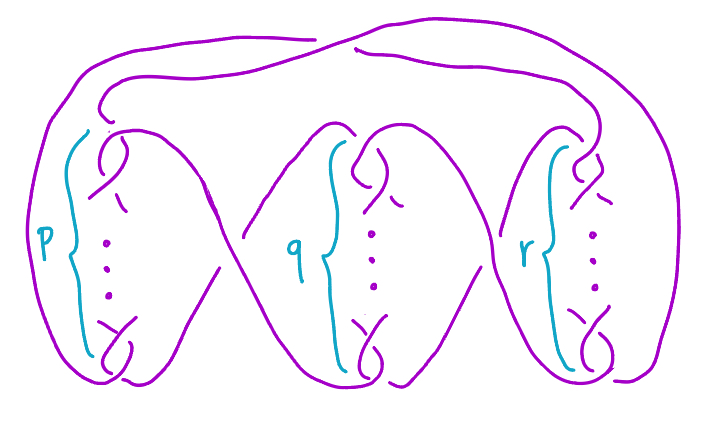}
    \caption{
    Coxeter link corresponding to a star-shaped Dynkin diagram with 3 legs with $p,q,r$ nodes. When $p=2$ this corresponds to the pretzel link $P(-2,q,r)$, but this is not true for $p\geq 3$.
    }
    \label{fig:coxeterlink}
\end{figure}

Among the HZ-factorisable cases, of particular interest is the  family of pretzel links $P(-2,3,n-3)$  \cite{Petrou,Petrou1}, which   at even $n=2j$ contains the knots $\mathcal{K}^{(3)}_{j-2,1,0}=\mathcal{K}^{(3)}_{j-3,0,1}$, while for odd $n=2j+1$ they are the two component links 
$\mathcal{L}^{(3)}_{j-1,1,0}=\sigma_2\otimes F_2^{j-1}\otimes \tilde{E}_3$.
The knot $12_{242}=P(-2,3,7)$, which is a member of this family at $n=10$, has the  interesting property that its
Alexander polynomial has  a real positive root, known as the famous Lehmer number 1.17628 \cite{Eriko}. These pretzel links are precisely the Coxeter links corresponding to $E_n$ type Dynkin diagrams, as indicated in Table~\ref{Tab:E_n} \cite{Eriko}.  Here the  exceptional group $E_n$  is  extended to  $n\geq 9$,  and its Dynkin diagrams  consist of  three parts with $3,2$ and $n-3$ nodes, respectively. 
\begin{table}[hbtp]
\centering
\begin{tabular}{llr}
\toprule
$n$ & \multicolumn{1}{l}{$E_n$ Dynkin diagram}&$P(-2,3,n-3)$\\
\midrule
4&\raisebox{-0.3\height}{\includegraphics[scale=0.7]{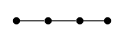}} ($E_4=A_4$)&$5_1=T(2,5)  $ \\
5&\hspace{0.35mm}\raisebox{-0.35\height}{\includegraphics[scale=0.67]{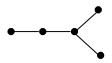}} ($E_5=D_5$) & $L7n1\{0\}^+=T(3,3,2,1)$\\
6&\raisebox{-0.5\height}{\centering\includegraphics[scale=0.7]{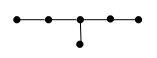}} & $8_{19}=T(3,4)$ \\
7&\hspace{-0.25mm}\raisebox{-0.55\height}{\includegraphics[scale=0.7]{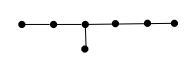}} & $L9n15\{0\}^+=T(3,4,2,1)$\\
8&\hspace{-0.05mm}\raisebox{-0.6\height}{\includegraphics[scale=0.7]{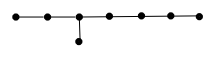}} & $10_{124}=T(3,5)$\\
9&\hspace{0.35mm}\raisebox{-0.5\height}{\includegraphics[scale=0.13]{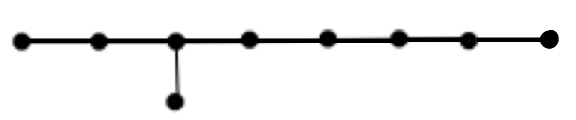}}&$L11n204\{0\}^+=T(3,5,2,1)$\\
10&\raisebox{-0.5\height}{\includegraphics[scale=0.13]{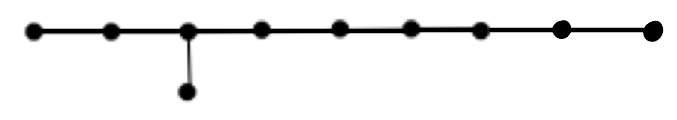}} & $12n_{242}$\\
\end{tabular}
\caption{$E_n$ Dynkin Diagrams and their Coxeter links $P(-2,3,n-3)$. For $n=4,5$ the diagrams  coincide with the ones of $A_4$ and $D_5$, respectively.
}\label{Tab:E_n}
\end{table}

This correspondence can be alternatively be observed via the HOMFLY--PT polynomial and its HZ transform.
It is possible to define the  HOMFLY--PT polynomial $P(L_n)$ of an ADE 
type Dynkin diagram $L_n$, viewed as a forest quiver,  via the recursive equation\footnote{This comes along with several normalisation conditions, as explained in \cite{Schwartz}, which are omitted here for simplicity.} \cite{Galashin, Schwartz} 
\be\label{PLn}
P(L_n)= \frac{z}{a}P(L_{n-1})+\frac{1}{a^2}P(L_{n-2}).
\ee
 For the generalised exceptional group 
 $P(E_n)$, which  is given in  \cite{Schwartz},  admits a factorisable HZ transform, as expected. In particular,  for $n\geq 4$,
\be\label{Z(En)}
Z(E_n)= \frac{\lambda(1-q^{-n-11}\lambda)(1- (-1)^n q^{-3n + 3}\lambda)}{(1-q^{-n+1}\lambda)(1-q^{-n-1}\lambda)(1-q^{-n-3}\lambda)(1-q^{-n-5}\lambda)},
\ee
which satisfies $Z(E_n)= Z(P(-2,3,n-3))$.  This can be expressed via the character expansion $Z(E_n) = \sum h^Q Z(\hat
S_Q)$, using (\ref{ZC3_1}) and the Racah coefficients 
\be
h^{[3]}= q^{n+2},\hskip 2mm
h^{[21]}=-\frac{q^{n-7}+q^{-n+7}}{q+q^{-1}},\hskip 2mm
h^{[111]}= (-q^{-1})^{n+2}.
\ee

It is also of interest to consider the HZ function corresponding to A and D type Dynkin diagrams. 
The HZ transform of  $P(A_n)$ for $n\ge 0$, is  evaluated as\footnote{Note that at $n=4$, corresponding to the 2-strand knot $5_1$, a factor $(1-\la q^{-9})$ in (\ref{Z(En)}) cancels between the denominator and numerator, yielding the same result as $Z(A_4)$.}
\be
Z(A_n)= \frac{\lambda(1-(-1)^n q^{-3n-3}\lambda)}{(1-q^{-n+1}\lambda)(1-q^{-n-1}\lambda)(1-q^{-n-3}\lambda)},
\ee
which is the same as $Z(T(2,n+1))$ for the 2-stranded torus links (these are knots for even $n$), which is factorised. 
These are obtained as the Coxeter links of  Fig.~\ref{fig:coxeterlink} with $p=3,q=1,r=n-2.$

The forest quiver polynomial for the $D_n$ series  is considered 
in \cite{Galashin}. For $D_2$ it is argued to be 
\be
P(D_2)= \biggl( \frac{z+z^{-1}}{a} -\frac{z^{-1}}{a^3}\biggr)^2.
\ee
This  corresponds to the HOMFLY--PT polynomial for the connected sum of two Hopf links $T(2,2)\#T(2,2)$, which is indeed the Coxeter link for $p=q=2$ and $r=0$. It does not yield a factorisable HZ function, 
but instead\footnote{The HZ transform is applied after the inclusion of the overall factor $(A-A^{-1})z^{-1}$.}  
\be\label{ZD2}
Z(D_2)=\frac{\lambda(1+(q^{-11}+ q^{-9}+q^{-7}+q^{-5})\lambda+ q^{-16}\lambda^2)}{(1-q^{-1}\lambda)(1-q^{-3}\lambda)(1-q^{-5}\lambda)(1-q^{-9}\lambda)}.
\ee
The Dynkin diagram for $D_3$ is the same as that of $A_3$, and consequently $P(D_3)=P(A_3)= (z^3+ 3 z + z^{-1}) a^{-3} - (z+ z^{-1}) a^{-5}$, which has HZ transform
\be\label{ZD3}
Z(D_3)= Z(A_3)= \frac{\lambda (1+q^{-12} \la)}{(1-q^{-2} \la)(1-q^{-4} \la)(1-q^{-6}\la)}.
\ee
This is the same as the HZ for the the link
$T(2,4)
=L4a1\{1\}$ (up to $q\to q^{-1}$). 
Since $D_5=E_5$, we have $Z(D_5)=Z(E_5)$, given by (\ref{Z(En)}).
Using the recursive relation $P(D_5)= \frac{z}{a}P(D_4)+\frac{1}{a^2}P(D_3)$, we can deduce 
\ba\label{ZD4}
Z(D_4)\hspace{-2mm}&=&\hspace{-2mm}\frac{\lambda(1+ 2 q^{-13} \lambda + 2 q^{-11} \lambda + q^{-24} \lambda^2)}{(1-q^{-3} \la)(1-q^{-5}\lambda)(1-q^{-7}\lambda)(1-q^{-9}\lambda)}\nonumber\\
\hspace{-2mm}&=&\hspace{-2mm} \frac{\lambda [ \frac{3}{2}(1+q^{-13}\lambda)(1+q^{-11}\lambda)-\frac{1}{2}(1-q^{-13}\lambda)(1-q^{-11}\lambda)]}{(1-q^{-3} \la)(1-q^{-5}\lambda)(1-q^{-7}\lambda)(1-q^{-9}\lambda)}.
\ea
The second line is the factorised form decomposition for $Z(D_4)$, 
 with fractional coefficients that sum to $1$. It is the same as the HZ of the 3-component link $L6n1\{0,1\}=T(3,3)$, given in Eq.~(29) of II. This is indeed the Coxeter link corresponding to $D_4$, for which $p=q=r=2$.
 Similarly,  $Z(D_6)$ can be evaluated from $P(D_6)= z a^{-1}P(D_5)+a^{-2}P(D_4)$, with $P(D_5)=P(E_5)$, to be
\be\label{ZD6}
Z(D_6) = \frac{\lambda( 1 + (q^{-19}+ q^{-17}+q^{-15} + q^{-13})\lambda + q^{-32}\lambda^2)}{
(1-q^{-5}\lambda)(1-q^{-7}\lambda)(1-q^{-9}\lambda)(1-q^{-11}\lambda)}.
\ee
This HZ function, which correspond to the Coxeter link $L8n3\{0,0\}$ with $p=q=2$ and $r=4$, admits the factorised form decomposition, similar to (\ref{ZD4}),
\ba
\hspace{-13mm}Z(D_6) \hspace{-2mm}&=&\hspace{-2mm} \frac{\lambda}{\mathcal{D}_3}\bigg(\frac{3}{4}[(1+q^{-19}\lambda)(1+q^{-13}\lambda)+(1+q^{-17}\lambda)(1+q^{-15}\lambda)]\nonumber\\
&&\hspace{6mm}-\frac{1}{4}[(1-q^{-19}\lambda)(1-q^{-13}\lambda)+(1-q^{-17}\lambda)(1-q^{-15}\lambda)]\bigg),
\ea
where $\mathcal{D}_3$ is the denominator of (\ref{ZD6}). 

\section{ Summary and discussion\label{sec:summary}}

In the present article we have shown that  the character expansion  is very effective in revealing the hidden structure of  the HOMFLY--PT polynomial and its HZ transform. This is due to the fact that the transform applies only to the Schur functions $\hat S_Q$, leaving the coefficients invariant, which helps illuminate the factorisability properties of the HZ function.
Namely, it provides sufficient conditions for HZ factorisation in terms of the Racah coefficients. These include that non-vanishing contributions should  come solely from  single hook Young diagrams. 

In a previous work \cite{Petrou1}, we have constructed special families of HZ-factoralisable knots and links, which are generated by full twists and Jucys--Murphy twists. Persistence of HZ factorisability under such twisting operators is clarified in the light of the  character expansion, as their $R-$matrices are diagonal and commute with each other and, hence, are the generators of commutative subalgebras. This result is used to rigorously construct more general families of HZ-factorisable knots, which often are hyperbolic extensions of torus knots.  Among them,  of particular interest is the family of pretzel links $P(-2,3,n-3)$, which are the Coxeter links corresponding to $E_n$ type singularities. 

As conjectured in \cite{Petrou1}, HZ factorisation is an important property from a physics point of view, as it implies a relation between the HOMFLY--PT and Kauffman polynomials that is equivalent to the vanishing
2-crosscap BPS invariants $\hat{N}_{g,Q}^{c=2}=0$. The number of BPS states of open topological strings, which correspond to knot invariants via the gauge/string duality, have been discussed  in \cite{Ooguri, Labastida}.
This property  of 2-crosscap BPS invariants for torus knots and their hyperbolic extensions, still lacks a  physical interpretation and, hence, deserves further  investigation. The HZ-factorisable families presented here are used in \cite{Petrou4} to   provide stronger evidence for the validity of this conjecture.

In more general cases,  for which the HZ function does not admit a factorised form,  we have shown that it still has an interesting structure, as  it can be  decomposed into a sum of  factorised terms. This is proven in the 3-strand case using the properties of the Racah coefficient, which is symmetric owing to the symmetry of the Young diagram. We have proposed an algorithm with which such a decomposition can be obtained, which is applicable for closed braids with up to 8 strands. 
The factorised form decomposition 
will be useful for the investigation of the real zeros of non-factorisable HZ functions. 
These yield Salem numbers, which can be thought of as Lyapunov exponents for a dynamical system. Further details about Salem numbers in relation to the HZ transform will be addressed elsewhere \cite{Petrou5,PetrouPHD}.

  \paragraph{Acknowledgments.}  We are very grateful to Andrew Lobb for illuminating discussions and his insightful comments about the Jucys--Muprhy twists. 
This work is supported by OIST funding  and by the collaboration fund between the University of Tokyo and OIST.  

\paragraph{Data Availability.} No data was used to support the findings of this study.

 \paragraph{Declaration of competing interest.} The authors have no competing interest to disclose. 
 
\newpage
\addcontentsline{toc}{section}{\numberline{}Appendix. The HZ function of Rolfsen knots in factorised form decomposition}
\section*{Appendix. The HZ function of Rolfsen knots in factorised form decomposition}

As explained in Sec.~\ref{sec:factform}, HZ functions can be grouped together according to their HZ type  $(\beta_0,...,\beta_m)$, where $\beta_i=e+2i$. 
Their numerators can be decomposed as a sum of factorised forms, in support of the Conjecture~\ref{conj:FactForm}.
 The table below 
 lists  the HZ type and factorised form decomposition for all knots $\mathcal{K}$ in the Rolfsen knot table with up to 8 crossings and selected knots with 9 crossings. These are obtained using the braid representations listed in \cite{KnotInfo} and the methods described in\footnote{The conventions adopted here are such that positive braids have positive exponents, matching the definitions in \cite{Petrou,Petrou1} (obtained by $q\to q^{-1}$ from the formulas derived using the conventions of this article).} Secs.~\ref{sec:CharExp} and~\ref{sec:factform}.  
  \setlength{\tabcolsep}{4pt} 
\renewcommand{\arraystretch}{0.5}
{\footnotesize 
\begin{longtable}{cll}
  \\  
    \hline
   \toprowspace
$\mathcal{K}$ & HZ type &\hspace{-2mm}$Z(\mathcal{K})$ in factorised form decomposition\\
  \hline
    \midrowspace
\endfirsthead 
\hline
\toprowspace
 $\mathcal{K}$ & HZ type &\hspace{-2mm}$Z(\mathcal{K})$ in factorised form decomposition\\
   \hline
    \midrowspace
  
\endhead

     \vspace{2mm}
       $3_1$  & $(1,3,5)$ &\hspace{-2mm}$[9]$\\\vspace{2mm}
        $4_1$  & $(-3,- 1,1,3)$ &\hspace{-2mm}$-[-5,5]+[-3,3]+[-1,1]$\\\vspace{2mm}
         $5_1$  & $(3,5,7)$ &\hspace{-2mm}$[15]$\\\vspace{2mm}
          $5_2$  & $(1,3,5,7)$ &\hspace{-2mm}$[3,13]$\\\vspace{2mm}
          $6_1$  & $(-3,-1,1,3,5)$ &\hspace{-2mm}$-[-5,1,9]+[-3,3,5]+[1,1,3]$\\\vspace{2mm}
          $6_2$  & $(-1,1,3,5)$ &\hspace{-2mm}$-[-3,11]+ [-1,9] + [3,5]$\\\vspace{2mm}
          $6_3$  & $(-3,-1,1,3)$ &\hspace{-2mm}$[-7,7]-[-5,5]+[-3,3]$\\\vspace{2mm}
          $7_1$  & $(5,7,9)$ &\hspace{-2mm}$ [21]$\\\vspace{2mm}
          $7_2$  & $(1,3,5,7,9)$ &\hspace{-2mm}$[3,5,17]-[5,9,11]+[7,9,9]$\\\vspace{2mm}
          $7_3$  & $(3,5,7,9)$ &\hspace{-2mm}$[5,19]+[9,15]-[11,13]$\\\vspace{2mm}
    $7_4$  & $(1,3,5,7,9)$ &\hspace{-2mm}$[3,5,17]+[3,9,13] -[5,9,11]$\\\vspace{2mm}
    $7_5$  & $(3,5,7,9)$ &\hspace{-2mm}$[5,19] - [7,17] + [9,15]$\\\vspace{2mm}
    $7_6$  & $(-1,1,3,5,7)$ &\hspace{-2mm}$[13,3,-1]-[11,7,-3]+[9,7,-1]$\\\vspace{2mm}
    $7_7$  & $(-5,-3,-1,1,3)$ &\hspace{-2mm}$- [-9,1,3]-[-9,-1,5]+ [-9,-3,7]-[-7,-3,5]$\\\vspace{2mm}
    &  &\hspace{-2mm}$+ [-3,-3,1]+ 2[-7,-1,3]$\\\vspace{2mm}
    $8_{1}$  & $(-3,-1,1,3,5,7)$ &\hspace{-2mm}$ [7, 7, -3, 1] - [13, -5, 3, 1] + [3, 5, 3, 1]$\\\vspace{2mm}
  $8_2$  & $(1,3,5,7)$ &\hspace{-2mm}$- [17,-1]+[15,1] +[11,5]$\\\vspace{2mm}
   $8_{3}$  & $(-5,-3,-1,1,3,5)$ &\hspace{-2mm}$-[9,1,-1,-9]+2[3,1,-1,-3]$\\\vspace{2mm}
    &  &\hspace{-2mm}$+[7,-7,1,-1] -[7,-7,3,-3]+[5,-5,3,-3]-[5,-5,1,-1]$\\\vspace{2mm}
  $8_4$  & $(-5,-3,-1,1,3)$ &\hspace{-2mm}$[-9, 1, 3] - [-11, -1, 7] + [-5, -1, 1] - 
 [-7, -1, 3] + [-3, -1, -1] $\\\vspace{2mm}
    $8_5$  & $(1,3,5,7)$ &\hspace{-2mm}$-[17, -1] + [15, 1] - [13, 3] + 
 [11, 5] + [9, 7] $\\\vspace{2mm}
    $8_6$  & $(-1,1,3,5,7)$ &\hspace{-2mm}$[13, 3, -1]- 
 [15, -1, 1]  + [5, 7, 3]- [11, -3, 7]  + [7, 9, -1] $\\\vspace{2mm}
  $8_7$  & $(-5,-3,-1,1)$ &\hspace{-2mm}$[-13, 5] - [-11, 3] + [-9, 1] - 
 [-7, -1] + [-3, -5]$\\\vspace{2mm}
  $8_8$  & $(-5,-3,-1,1,3)$ &\hspace{-2mm}$[-11, 7, -1] - [-9, -1, 5] + [-7, -3, 5] - 
 [-11, 1, 5] + [-11, 3, 3] $\\\vspace{2mm}
  $8_9$  & $(-3,-1,1,3)$ &\hspace{-2mm}$- [-9,9] + [-7,7]- [-5,5] + 2 [-3,3] $\\\vspace{2mm}
  $8_{10}$  & $(-5,-3,-1,1)$ &\hspace{-2mm}$[-13, 5] - [-11, 3] + 2 [-9, 1] - 
 [-7, -1]$\\\vspace{2mm}
  $8_{11}$  & $(-1,1,3,5,7)$ &\hspace{-2mm}$[-1, 3, 13] - [-3, 3, 15] + [-1, 7, 9] - 
 [1, 3, 11] + [3, 3, 9] $\\\vspace{2mm}
    $8_{12}$  & $(-5,-3,-1,1,3,5)$ &\hspace{-2mm}Eq.~(\ref{8_12})\\\vspace{2mm}
    $8_{13}$  & $(-5,-3,-1,1,3)$ &\hspace{-2mm}$[-11, -1, 7] - [-9, -1, 5] + [-7, -1, 3] - 
 [-5, -5, 5] + [-5, -3, 3] $\\\vspace{2mm}
    $8_{14}$  & $(-1,1,3,5,7)$ &\hspace{-2mm}$2 [-1, 3, 13] - [-3, 3, 15] + [3, 3, 9] - 
 [1, 3, 11] $\\\vspace{2mm}
    $8_{15}$  & $(3,5,7,9,11)$ &\hspace{-2mm}$[5, 15, 15] - [7, 7, 21] + [7, 9, 19] - 
 [7, 11, 17] + [5, 11, 19] $\\\vspace{2mm}
    $8_{16}$  & $(-5,-3,-1,1)$ &\hspace{-2mm}$[-13,5]-2[-11,3]+2[-9,1]-[-7,-1]+[-5,-3]$\\\vspace{2mm}
     $8_{17}$  & $(-3,-1,1,3)$ &\hspace{-2mm}$- [-9,9] + 2 [-7,7] - 2 [-5,5]+ 2 [-3,3]$\\\vspace{2mm}
    $8_{18}$  & $(-3,-1,1,3)$ &\hspace{-2mm}$ -[-9,9]+ 3 [-7,7] -2[-5,5] + 2[-3,3]-[-1,1] $\\\vspace{2mm}
    $8_{19}$  & $(5,7,9,11)$ &\hspace{-2mm}$[15,17] $\\\vspace{2mm}
    $8_{20}$  & $(-5,-3,-1,1)$ &\hspace{-2mm}$[-11,3]$\\\vspace{2mm}
     $8_{21}$  & $(1,3,5,7)$ &\hspace{-2mm}$ -[1,15] + [3,13]+ [7,9] $\\\vspace{2mm}
     $9_1$ & $(7,9,11)$ &\hspace{-2mm}$[27]$\\\vspace{2mm}
     $9_2$ & $(1,3,5,7,9,11)$  &\hspace{-2mm}$[5,21, 3, 7] + [5,11, 9, 11] - [5,15, 5, 11]$\\\vspace{2mm}
     $9_3$ & $(5,7,9,11)$ & \hspace{-2mm}$[25, 7] - [13, 19] + [21, 11]$\\\vspace{2mm}
     $9_4$ & $(3,5,7,9,11)$ & \hspace{-2mm}$[3, 9, 23] - [11, 11, 13] + [5, 11, 19] - 
 [3, 15, 17] + [9, 11, 15]$\\\vspace{2mm}
 $9_5$ & $(1,3,5,7,9,11)$&\hspace{-2mm}$[21, 3, 7,5] + [17, 3, 11,5] - [15, 5, 11,5]$\\\vspace{2mm}
 $9_6$& $(5,7,9,11)$ &\hspace{-2mm}$[25, 7] - [23, 9] + [21, 11] - [13, 19] + [15, 17]$\\\vspace{2mm}
 $9_7$ & $(3,5,7,9,11)$ &\hspace{-2mm}$[3, 9, 23] - [7, 11, 17] + [5, 11, 19] - 
 [3, 11, 21] + [9, 11, 15]$\\\vspace{2mm}
 $9_{8}$ &$(-3,-1,1,3,5,7)$&\hspace{-2mm}$-[11,7,1,-7]+[9,5,3,-5]+[13,1,-1,-1]
+ 2[11,-3,5,-1]$\\\vspace{2mm}
 &&\hspace{-2mm}$-2[11,-3,3,1]+([9,-1,3,1]-[9,-1,5,-1]$\\\vspace{2mm}
 &&\hspace{-2mm}$+[5,3,5,-1]- [5,3,3,1])$\\\vspace{2mm}
 $9_9$ & $(5,7,9,11)$&\hspace{-2mm}$[25, 7] - [23, 9] + 2 [21, 11] -  [13, 19]$\\\vspace{2mm}
 $9_{10}$ &$(3,5,7,9,11)$& \hspace{-2mm}$[7, 9, 19] - [7, 11, 17] + [3, 9, 23] - 
 [7, 13, 15] + [5, 11, 19] $\\\vspace{2mm}
 &&\hspace{-2mm}$- [3, 15, 17] +  [5, 15, 15]$\\\vspace{2mm}
 $9_{11}$& $(-9,-7,-5,-3,-1)$&\hspace{-2mm}$[-19, -1, -5] - [-1, -7, -17] + [-15, -5, -5] -  [-13, -7, -5]$\\\vspace{2mm}
 &&\hspace{-2mm}$ + [-15, -7, -3] - 
 [-23, -3, 1] + [-23, -1, -1]$\\\vspace{2mm}
   $9_{12}$ &$(-1,1,3,5,7,9)$&\hspace{-2mm}$-[15,11,-3,1] + 2[13,9,3,-1]+[17,7,1,-1]-[13,11,1,-1]$\\\vspace{2mm}
 &&\hspace{-2mm}$+ ([ 17,-1,9,-1]-[17,-1,7,1]+[11,5,7,1]-[11,5,9,-1]) $\\\vspace{2mm}
  $9_{13}$ & $(3,5,7,9,11)$&\hspace{-2mm}$[5, 7, 23] - [7, 7, 21] + 2 [7, 9, 19]$\\\vspace{2mm}
  &&\hspace{-2mm}$- 
 [7, 11, 17] + [5, 15, 15] - [7, 13, 15]$\\\vspace{2mm}
   $9_{15}$ &\hspace{-2mm}$(-9,-7,-5,-3,-1,1)$&\hspace{-1mm}$[-17,-3,-3,-1] + 2[-13,-9,-3,1]-[-15,-11,3,-1]$\\\vspace{2mm}
 &&\hspace{-2mm}$-[-11,-9,-3,-1] +([- 17,3,-9,-1] -[-17,3,-7,-3]
$\\\vspace{2mm}
 &&\hspace{-2mm}$ +[-13,-1,-7,-3]-[-13,-1,-9,-1])+ 2([-19,-7,-3,5]$\\\vspace{2mm}
 &&\hspace{-2mm}$-[-19,-7,-1,3]+[-15,-11,-1,3]-[-15,-11,-3,5]) $\\\vspace{2mm}
  $9_{16}$ & $(5,7,9,11)$& \hspace{-2mm}$[25, 7] - 2 [23, 9] + 2 [21, 11] - [13, 19] + 
 [15, 17]$\\\vspace{2mm}
  $9_{17}$ & $(-3,-1,1,3,5)$&\hspace{-2mm}$[-7, -1, 13] - 2 [-7, 1, 11] + [-5, 1, 9] - 
 2 [-3, 1, 7] + $\\\vspace{2mm}
 &&\hspace{-2mm}$[1, 1, 3] + [-1, 1, 5] + 
 [-3, 3, 5]$\\\vspace{2mm}
  $9_{18}$ & $(3,5,7,9,11)$&\hspace{-2mm}$[5, 7, 23] - [7, 7, 21] + 2 [7, 9, 19] - 
 2 [7, 11, 17] + [9, 11, 15]$\\\vspace{2mm}
 &&\hspace{-2mm}$ - [7, 9, 19] +  [5, 11, 19]$\\\vspace{2mm}
  $9_{20}$ & $(1,3,5,7,9)$&\hspace{-2mm}$[1, 9, 15] - [1, 7, 17] + [1, 5, 19] - 
 [3, 5, 17] + [5, 5, 15]$\\\vspace{2mm}
 &&\hspace{-2mm}$ - [-1, 13, 13] + 
 [1, 11, 13]$\\\vspace{2mm}
   $9_{22}$ & $(-3,-1,1,3,5)$&\hspace{-2mm}$ [-5, -3, 13] - [-7, 1, 11] + [-3, -1, 9] - 
 [-3, -3, 11] + [-7, 3, 9]$\\\vspace{2mm}
 &&\hspace{-2mm}$  - [-3, 1, 7] + 
 [-1, 1, 5]- [-7, 5, 7] + [-1, 1, 5]$\\\vspace{2mm}
    $9_{23}$ & $(3,5,7,9,11)$& \hspace{-2mm}$ [5, 7, 23] - 2 [7, 7, 21] + 2 [7, 9, 19] - 
 [7, 11, 17] + [5, 15, 15]$\\\vspace{2mm}
 $9_{24}$ & $(-3,-1,1,3,5)$&\hspace{-2mm}$[-3, -3, 11] - [-5, 1, 9] + [-3, 1, 7] - 
 [-9, 5, 9] + [-7, 5, 7]$\\\vspace{2mm}
 &&\hspace{-2mm}$ - [-5, 5, 5] + 
 [-7, 3, 9]$\\\vspace{2mm}
 $9_{26}$ & $(-7,-5,-3,-1,1)$&\hspace{-2mm}$[-11, -9, 5] - [-15, -1, 1] + 2 [-13, -3, 1] - 
 2 [-11, -7, 3]$\\\vspace{2mm}
 &&\hspace{-2mm}$ + [-9, -7, 1] - 
 [-11, -3, -1] + [-7, -5, -3]$\\\vspace{2mm}
  $9_{27}$ & $(-3,-1,1,3,5)$&\hspace{-2mm}$[-3, -3, 11] - [-5, 1, 9] + 2 [-7, 5, 7] - 
 [-9, 5, 9] + [1, 1, 3]  $\\\vspace{2mm}
 &&\hspace{-2mm}$- [-5, 5, 5]+ 
 [-3, 3, 5] - [-1, 1, 5]$\\\vspace{2mm}
   $9_{28}$ & $(-1,1,3,5,7)$&\hspace{-2mm}$[-1, 3, 13] - [-3, 3, 15] + [-5, 7, 13] - 
 [-3, 7, 11] + [3, 3, 9] $\\\vspace{2mm}
 &&\hspace{-2mm}$- [1, 3, 11] + 
 2 [-1, 7, 9] - [1, 7, 7] $\\\vspace{2mm}
  $9_{29}$ & $(-5,-3,-1,1,3)$&\hspace{-2mm}$[3, 5, -13] - [-11, -1, 7] + 2 [-9, 1, 3] - 
 [-11, 3, 3] + [-9, -1, 5]  $\\\vspace{2mm}
 &&\hspace{-2mm}$- 2 [-7, -1, 3] + 
 [-5, -1, 1]$\\\vspace{2mm}
 $9_{30}$ & $(-5,-3,-1,1,3)$&\hspace{-2mm}$[-11, -1, 7] - 2 [-9, -1, 5] + 2 [-7, -1, 3] - 
 [-9, -5, 9]$\\\vspace{2mm}
 &&\hspace{-2mm}$ + [-9, -3, 7] - [-5, -5, 5] + 
 [-5, -3, 3]$\\\vspace{2mm}
  $9_{31}$ & $(-1,1,3,5,7)$&\hspace{-2mm}$3 [-1, 3, 13] - [-3, 3, 15] + [-5, 9, 11] - 
 2 [1, 3, 11] + [9, 9, -3] $\\\vspace{2mm}
 &&\hspace{-2mm}$ - [-3, 7, 11] + 
 [3, 3, 9] - [-3, 9, 9]$\\\vspace{2mm}
   $9_{32}$ & $(-7,-5,-3,-1,1)$& \hspace{-2mm}$[-5, 5, -15] - [-15, -1, 1] + 
 2 [-13, -3, 1] - 3 [-11, -7, 3]$\\\vspace{2mm}
 &&\hspace{-2mm}$ + 
 2 [-9, -7, 1] - [-15, -1, 1] + [-9, -7, 1]$\\\vspace{2mm}
 $9_{33}$ & $(-5,-3,-1,1,3)$& \hspace{-2mm}$[-11, 3, 3] - 2 [-9, -5, 9] + [-7, -7, 9] - 
 [-9, -1, 5] + [-7, -1, 3] $\\\vspace{2mm}
 &&\hspace{-2mm}$- [-9, 3, 1] + 
 2 [-9, -3, 7]  - [5, 5, -15] + [-15, 3, 7]$\\\vspace{2mm}
   $9_{42}$ & $(-3,-1,1,3,5)$& \hspace{-2mm}$- [-7,7,5] + [-3,3,5] + [-1,1,5]$\\\vspace{2mm}
     $9_{49}$ & $(3,5,7,9,11)$& \hspace{-2mm}$2[19,5,11]-[17,7,11]+[15,9,11]-[13,11,11]$\\
          \bottomrule
   \label{tab:listFFD}
\end{longtable}
}

The knot $8_{12}$, which is of HZ type $(-5,-3,-1,1,3,5)$, has factorised form decomposition 
\ba\label{8_12}
Z(8_{12}) &=& -[-9,-5,5,9]+ [-7,7,-3,3]+[-3,3,-1,1]\nonumber\\
&&+(A_1+A_3-A_6-3(A_4-A_5-A_2)).
\ea 
Here  $A_n$ denote the quadruple corrections  and are defined as  
\ba\label{An}
A_1&:=& [-9,5,-5,9]-[-9,5,-3,7]+[-7,3,-3,7]-[-7,3,-5,9]\nonumber\\
A_2&:=&[-7,5,-5,7]-[-7,5,-3,5]+[-5,3,-3,5]-[-5,3,-5,7]\nonumber\\
A_3&:=&[-7,3,-3,7]-[-7,3,-1,5]+[-5,1,-1,5]-[-5,1,-3,7]\nonumber\\
A_4&:=&[-5,3,-3,5]-[-5,3,-1,3]+[-3,1,-1,3]-[-3,1,-3,5]\nonumber\\
A_5&:=&[-5,5,-3,3]-[-5,5,-1,1]+[-3,3,-1,1]-[-3,3,-3,3]\nonumber\\
A_6&:=&[-3,3,-3,3]-[-3,3,-1,1]+[-1,1,-1,1]-[-1,1,-3,3].\nonumber
\ea

\vskip2mm
 \noindent{Remark.} In the HZ functions  for the knots $5_2$, $8_1$, $8_{14}$, $9_2$, $9_5$, $9_{42}$ and $9_{49}$ presented in the above table,  one exponent, say $\alpha_*$,  appears in all the factorised terms of the numerator and also in the denominator. Although  the corresponding common factor $(1-\la q^{\alpha_*}) $ could be canceled out, we include it here to reflect the braid index $\iota$ of the knot, which, by doing this, coincides with $m$. It is noteworthy, however, that the knots $9_{42}$ and $9_{49}$ are instances of non-sharp Morton-Franks-Williams inequality ($\iota>m$) and hence  their HOMFLY--PT polynomial is in fact oblivious of the missing exponent. This can be understood from the HZ transform,  due to the fact that such an exponent 
lies at the edge of the range of the HZ type, resulting in an ambiguity on which of its ends  it should be added. This can be clarified, however, via the method described in \cite{Petrou1}, and the corresponding exponent has been included in the decompositions presented above. 


\newpage
\begin{thebibliography}{99}
\bibitem{Witten0}
E. Witten, Quantum field theory and the Jones polynomial, Commun. Math. Phys. 121 (1989) 351-399.
\bibitem{Reshetikhin}
N. Yu. Reshetikhin and V. G. Turaev, Ribbon graphs and their invariants derived from quantum groups, Comm. Math. Phys. 127 (1990) 1. 
\bibitem{MorozovCH}
A. Mironov, A. Morozov and And. Morozov, Character expansion for HOMFLY polynomials I. Integrability and difference equations, 
Strings, Gauge Fields, and the Geometry Behind, pp.101(2012), World Scientific. arXiv:1112.5754.
\bibitem{MorozovCHII}
A. Mironov, A. Morozov and And. Morozov, Character expansion for HOMFLY polynomials II: Fundamental representation. Up to five strands in braid. JHEP 03 (2012) 034, arXiv: 1112.2654.
\bibitem{Anokhina}
A. Anokina, A. Mironov, A. Morozov and And. Morozov,
Racah coefficients and extended
HOMFLY polynomials for all 5-,6-,7-strand braids, Nucl. Phys. B868 (2013) 271. arXiv: 1207.0279.
\bibitem{And}
And. Morozov, Multistrand eigenvalue conjecture and Racah symmetries, arXiv: 2212.01289.


\bibitem{Petrou}
A. Petrou and S. Hikami, The Harer--Zagier transform of the HOMFLY--PT polynomial for families of twisted hyperbolic knots, J. Phys. A: Math. Theor. 57 (2024) 205204. \href{https://iopscience.iop.org/article/10.1088/1751-8121/ad421b}{DOI:10.1088/1751-8121/ad421b} (denoted as I).
\bibitem{Petrou1}
A. Petrou and S. Hikami, 
The HOMFLY--PT polynomial and HZ factorisation,  {\href{https://arxiv.org/abs/2412.04933}{arXiv:2412.04933}} (denoted as II).
\bibitem{Petrou4}
A. Petrou and S. Hikami, A relation between the HOMFLY-PT and Kauffman polynomials via characters, {\href{https://arxiv.org/abs/2603.03628}{arXiv:2603.03628}}.
\bibitem{Petrou5}
A. Petrou and S. Hikami,
Zeros of knot polynomials, Salem numbers and Dualities, In preparation.
\bibitem{KauffmanLins}
L.Kauffman and S. Lins, Temperley-Lieb recoupling theory and invariants of 3-manifold, Annals of Math Studies 134 (1994).
\bibitem{Kirillov}
A.N. Kirillov and N.Yu. Reshetikhin, Representations of the algebra $U_q(sl(2))$, q-orthogonal polynomials and invariants of links, in "Infinite dimensional Lie algebras and groups, edit. by V. G. Kac, in Advanced series in mathematical physics Vol.7, (1989) World Scientific. 
\bibitem{ItzyksonZuber}
C. Itzykson and J.-B. Zuber, The planar approximation, II. J. Math. Phys. 21 (1980) 411.


\bibitem{Bar-Natan}
Knot Atlas, http://katlas.org/wiki/Main Page (by D. Bar-Natan).

\bibitem{BrezinHikami}
E. Br\'ezin and S. Hikami, Random matrix theory with an external source, Springer Brief in Mathematical Physics 19 (2016), Springer.



\bibitem{KnotInfo}
KnotInfo: https://knotinfo.math.indiana.edu/index.php
\bibitem{Murasugi}
K. Murasugi, Knot theory and its applications,  Birkh\"auser (1996).

\bibitem{Galashin}
P. Galashin and T. Lam,
Plabic links, quivers, and skein relations,  Combinatorics 7 (2024) 431. arXiv: 2208.01175.
\bibitem{Schwartz}
A. Schwartz, The HOMFLY polynomial of a forest quiver, 
Seminaire Lotharingien de combinatoire 93B (2025) 146,  Proceedings of the 37th conference on formal power series and algebraic combinatorics (Sapporo). arXiv: 2410.00399.
\bibitem{Eriko}
E. Hironaka, Lehmer's problem, Mckay's correspondence, and  2,3,7. Topics in algebraic and noncommutative geometry, Contem. Math. 324 (2001) 123.
arXiv: math/0204040.
\bibitem{ACampo}
N. A’Campo, Le groupe de monodromie du déploiement des singularités isolées de courbes planes II, Proceedings of the International Congress of Mathematicians, Vol. 1 (1974) 375.


\bibitem{Morozov}
A. Morozov, A. Popolitov and Sh. Shakirov, Harer--Zagier formulas for knot matrix model. Phys. Lett. B818, (2021) 136370.
arXiv: 2102.11187.
\bibitem{Ooguri}
H. Ooguri and C. Vafa, Knot invariants and topological strings, Nucl. Phys.B 577,  (2000) 419.  arXiv.hep-th/9912123.
\bibitem{Labastida}
J.M.F. Labastida, M. Marino and C. Vafa, Knots, links and branes at large N, JHEP 11, (2000) 007, arXiv.hep-th/0010102.

\bibitem{Kauffman}L. H. Kauffman, Knots and physics, Vol. 1. World scientific (2001).
\bibitem{PetrouPHD}A. Petrou, Deciphering knot polynomials via the Harer-Zagier transform, Okinawa Institute of Science and Technology, PhD dissertation (2026).

\bibitem{HZ}
J. Harer and D. Zagier, Euler characteristics of the moduli space of curves. Inv. Math. 85 (1986) 457.


 \bibitem{MorozovU}
A. Mironov, Mkrtchyan and A. Morozov, On universal knot polynomials, JHEP 02 (2016) 078. arXiv:1510.05884.


\end{thebibliography}
\end{document}